\documentclass[aps,prx,reprint,superscriptaddress,amsmath,amssymb]{revtex4-2}
\usepackage{hyperref}
\hypersetup{
colorlinks=true,
linkcolor=red,
citecolor=red,}
\usepackage{blindtext}
\usepackage{amsmath}
\usepackage{amssymb}
\usepackage{amsthm}
\usepackage{dsfont}
\usepackage{bm}        % For bold symbols
\usepackage{mathrsfs}  % For script symbols
\usepackage[dvipsnames]{xcolor}
\usepackage{titlesec}
\usepackage{tikz}
\usetikzlibrary{calc}
\usepackage{graphicx}
\usepackage{array}      % For table formatting
\usepackage{float}
\usepackage{ulem}
\usepackage{subcaption}
\usepackage{appendix}
\newcommand{\MM}{\mathsf{M}}
\newcommand{\pp}{\mathsf{P}}
\newcommand{\e}{\mathsf{E}}
\newcommand{\rank} {\operatorname{rank}}
\usepackage{enumitem}

\newtheorem{theorem}{Theorem}
\newtheorem{lemma}[theorem]{Lemma} % Number lemma with theorems
\newtheorem{corollary}[theorem]{Corollary} % Number corollary with theorems
\theoremstyle{definition} % Definition style: upright text
\newtheorem{proposition}[theorem]{Proposition} % Propositions share numbering with theorems
\newtheorem{definition}{Definition}
\usepackage{algcompatible}
\usepackage{algorithm}

\titlespacing*{\subsection}{0pt}{\baselineskip}{\baselineskip}

\begin{document}
\title{Complexity of Contextuality}
\author{Theodoros Yianni}
\affiliation{Department of Computer Science, Royal Holloway, University of London, United Kingdom}
\author{Farid Shahandeh}
\affiliation{Department of Computer Science, Royal Holloway, University of London, United Kingdom}

\begin{abstract}
Generalized contextuality is a 
hallmark of nonclassical theories like quantum mechanics.
Yet, three fundamental computational problems concerning its decidability and complexity remain open.
First, determining the complexity of deciding if a theory admits a noncontextual ontological model;
Second, determining the complexity of deciding if such a model is possible for a specific dimension $k$;
Third, efficiently computing the smallest such model when it exists, given that finding the smallest ontological model is NP-hard.
We address the second problem by presenting an algorithm derived from a geometric formulation and its reduction to the intermediate simplex problem in computational geometry.
We find that the complexity of deciding the existence of a noncontextual ontological model of dimension $k$ is at least exponential in the dimension of the theory and at most exponential in $k$.
This, in turn, implies that computing the smallest noncontextual ontological model is inefficient in general.
Finally, we demonstrate the fundamental difference between finding the smallest noncontextual ontological model and the smallest ontological model using an explicit example wherein the respective minimum ontic sizes are five and four. 
\end{abstract}

\maketitle
\section{Introduction}
Quantum contextuality is the property that a quantum experiment cannot be modelled by assigning fixed hidden variable representations to its states and outcome effects \cite{Koc,Mermin1993,Spekkens_2005}.
This peculiar property is considered to be an important resource for quantum advantages in many tasks such as quantum computation \cite{computation,shahandeh2021advantage,Schmid2022Stabilizer}, communication \cite{communication}, optimal state discrimination \cite{Schmidt2018QMESD, ourpaper}, and quantum cloning \cite{ourpaper,lostaglio2020contextual}.

Determining the size of the ontological (hidden-variable) models for quantum theory, and more broadly, for any generalized probabilistic theory (GPT), is an active field of research \cite{Korzekwa_2021, Havl_ek_2020, Montina_2008, HARDY2004267}.
Indeed, the principle of parsimony suggests that when searching for an ontological explanation of the data, the simplest model with the fewest ontic states is preferred.
In particular, a large number of ontic states necessary for a model compared to the quantum dimensionality indicates quantum advantage, an example of which can be identified in the context of classical versus quantum communication~\cite{Shitov2019}.
Notably, when a noncontextual model exists, an upper bound on its size is the square of the GPT dimension~\cite{Gitton_2022}.
The value of such bounds is that they can be used to explain nonclassical advantages of protocols under mild assumptions~\cite{Galvao2003,Montina_2008,ourpaper}.
They also enable proofs of contextuality in cases where the minimum required ontic size exceeds the minimum size of noncontextual models.
While it is possible to find the smallest noncontextual model, the currently available method to achieve this is to optimize over all solutions of a linear program \cite{Selby_2024}. 
However, this is challenging because there may be an infinite number of such solutions, rendering the optimization highly inefficient.

With the introduction of generalized contextuality \cite{Spekkens_2005}, there has been significant progress in theoretical and experimental methods to detect it~\cite{linearprogram,Schmid_2024,Gitton_2022}.
These methods generally rely on linear programs, such as the one presented in Ref.~\cite{Selby_2024}, which detects contextuality in arbitrary GPTs.
Alternatively, by considering nonnegative matrix factorizations (NMFs) of the matrix of conditional outcome probability evaluations (COPE), one arrives at a criterion for noncontextuality, as introduced in Ref.~\cite{ourpaper}. More specifically, noncontextuality implies that the ranks of the COPE decomposition factors must be equal, the violation of which is called the \textit{rank separation}. This allows us to detect contextuality by studying bounds on the ranks of the NMF factors.
Unlike other methods, the rank separation approach does not require explicit analysis of operational equivalences and operational identities.

In this work, we present a method for finding the smallest noncontextual model of the prepare-measure experiments, along with its time complexity.
We achieve this by leveraging the rank separation criterion and providing a new geometric interpretation of contextuality.
This geometric perspective is used to reduce the problem of finding the smallest noncontextual model to an NMF problem. The latter has a known complexity and can be solved with existing algorithms, such as Moitra's almost optimal algorithm~\cite{moitra}. In particular, we show that for COPE matrices of fixed rank, the existence of a noncontextual model can be decided in polynomial time. 
However, for $m \times n$ matrices of arbitrary rank $r$, the complexity is at least $(nm)^{O(r)}$ and at most $\operatorname{poly}(b,m,n)^{O(k^2)}$ wherein $b$ is the bit length of the matrix entries and $k$ is the size of the ontological model.

It is already known that the smallest (not necessarily noncontextual) ontological model of a prepare-measure scenario is given by the NMF of its COPE matrix~\cite{ourpaper},
the inner dimension $k$ of which determines the cardinality of the ontic space. 
It is also known that performing an NMF of the smallest inner dimension is NP-hard~\cite{vavasis}, i.e., computing the smallest ontological model for a prepare-measure scenario is NP-hard in the worst case.
An interesting question is thus whether the smallest ontological model and the smallest noncontextual ontological model coincide or a size gap between the two is possible.
We present an explicit example of a COPE matrix where such a gap exists. 
This means that, in general, the two classes of ontological models are different in size,
which opens up the possibility that different ontic sizes may be responsible for different types of quantum advantages.

This paper is organized as follows. In Sec.~\ref{section 2} we give an overview of the preliminary concepts and definitions, including COPE matrix formalism and its underlying models. In Sec.~\ref{section 3}, we reformulate the problem of finding the smallest noncontextual model as an NMF problem and give an upper bound for its complexity.
In Sec.~\ref{sec:comparison}, we clarify the relationship between our approach and the existing embeddability results.
Finally, in Sec.~\ref{Ontic gap} we demonstrate an example where there is a size gap between the smallest ontological and noncontextual ontological models.
Discussions and conclusions are presented in Sec.~\ref{section 5}.

In what follows, the following notation is adopted.
We use sans-serif capital Latin letters ($\mathsf{A}$) to denote experimental procedures and outcomes. The generic real matrices are denoted by capital Latin letters ($A$), while bold capital Latin letters ($\mathbf{A}$) are used to denote nonnegative matrices. 
The symbol $\mathds{1}$ represents the identity matrix.
Vectors are denoted by bold italic lowercase Latin letters ($\bm{a}$); vector spaces are denoted by Fraktur capital Latin letters ($\mathfrak{A}$); $\mathbf{1}$ is used for the vector of ones whose length is determined by the context if not explicitly mentioned.
In working with geometrical objects, calligraphic capital Latin letters ($\mathcal{A}$) are used to represent polytopes.

\section{Preliminaries} \label{section 2}

\subsection{Conditional Outcome Probability Evaluations Matrix}

The primitive elements of any physical experiment are the propositions describing each operational procedure.
Those associated with preparing the system are called \textit{preparations} and denoted by $\pp$.
For example, ``Turn the laser on AND adjust its intensity as such AND align its beam as such'' is a preparation of a photonic system.
Similarly, propositions such as ``Turn the detector on AND align it as such AND open its aperture'' associated with measurements of the system are called \textit{measurements}, and outcomes of measurements are \textit{events}.
We denote these by $\MM$ and $\e$, respectively.
For all practical purposes, we can assume that the sets of preparations $\{\pp_i\}$, measurements $\{\MM_j\}$, and events $\{\e_k\}$ are finite.
In the most general scenario, we can associate conditional probabilities to events as $p(\e_k|\pp_i,\MM_j)$ indicating the probability of event $\e_k$ in the measurement $\MM_j$ following the preparation $\pp_i$.
This data can be put into a matrix where each row represents a particular event, and each column represents a preparation,

\begin{equation}  \label{eq: C^a}
\begin{split}
    \mathbf{C}^a :&= 
\begin{pmatrix}
\begin{array}{c}
   \begin{matrix}
p(\e_{1}|\pp_1,\MM_1) & \cdots & p(\e_{1}|\pp_n,\MM_1) \\
        \vdots &   & \vdots  \\
        p(\e_{m_1}|\pp_1,\MM_1) & \cdots & p(\e_{m_1}|\pp_n,\MM_1)
   \end{matrix}\\ 
   \hline
    \vdots \\
   \hline
   \begin{matrix}
        p(\e_{1}|\pp_1,\MM_l) & \cdots & p(\e_{1}|\pp_n,\MM_l)\\
        \vdots & & \vdots \\
        p(\e_{m_l}|\pp_1,\MM_l) & \cdots & p(\e_{m_l}|\pp_n,\MM_l)
   \end{matrix} \\
\end{array}    
\end{pmatrix}\\
&=\begin{pmatrix}
    \mathbf{C}^{a1}\\
    \vdots \\
    \mathbf{C}^{a,l}
\end{pmatrix},
\end{split}
\end{equation}
where we have assumed $l$ measurements with $m_j$ outcomes each, and $n$ preparations. Each submatrix $\mathbf{C}^{ai}$ corresponds to the measurement $\MM_i$. We call the matrix of \textit{conditional outcome probability evaluations} (COPE)~\cite{ourpaper}.
We can also represent the same information about the system with another COPE as follows. Construct the measurement procedure $\mathsf{M}$ by choosing a measurement from $\{\MM_i\}$ uniformly at random. The set of events for $\MM$ is the union of the sets of events of the previous measurements.  
Thus $\mathsf{M}$ has $m=\sum_{j=1}^l m_j$ possible events, each with probability $p(\mathsf{E}_k|\pp_i)$, where we have dropped the dependence on the measurement for brevity. The COPE matrix $\mathbf{C}^{b}$ obtained from using only the measurement $\MM$ and the same set of preparations is given by,
\begin{gather}\label{eq:COPE_def}
\mathbf{C}^{b} := 
\begin{pmatrix}
        p(\e_1|\pp_1) & \cdots & p(\e_1|\pp_n) \\
        \vdots &   & \vdots  \\
        p(\e_m|\pp_1) & \cdots & p(\e_m|\pp_n)
\end{pmatrix}.
\end{gather}
This matrix is a rescaling of $\mathbf{C}^{a}$ making it column-stochastic as shown in Ref.~\cite{ourpaper}. 

\subsection{Ontological Model} \label{subsec:OM}

One may want to construct an ontological model that explains the statistics of an experiment on a system.
Such a model consists of a set of ontic states indexed by $\lambda$, that represent the different mutually exclusive definite states the system can be in. 
The set of all ontic states is called the ontic space, which is conventionally denoted by $\Lambda$ and its cardinality is $|\Lambda|$.

In general, a preparation $\pp$ is modelled in an ontological model as some epistemic state $\bm{\mu}^{(\pp)}$.
An epistemic state is a stochastic vector $(\mu_\lambda)_\lambda$, where $\mu_\lambda$ is the probability of the system being in the ontic state $\lambda$.
A measurement, on the other hand, is given by a collection of response functions $\{\bm{\xi}^{i}\}$ where each $\bm{\xi}^{i}$ corresponds to a different event.
The response functions for a measurement are vectors with entries in $[0,1]$. The response functions sum to the unit response function $\mathbf{1}$.

The probability of the event $\e_k$ corresponding to the response function $\bm{\xi}^{k}$ of $\MM$ given that the preparation $\pp$ gives $\lambda$ as a complete ontic description of the state of the system is,
\begin{equation}
    p(\e_k|\pp)=\langle \bm{\xi}^{k},\bm{\delta}^{\lambda}\rangle.
\end{equation}
Here, $\bm{\delta}^{\lambda}=(\delta^{\lambda}_{\lambda'})_{\lambda'}$ with $\delta^i_j$ being the Kronecker delta, and $\langle\cdot,\cdot\rangle$ is the inner product over $\mathbb{R}^{|\Lambda|}$.
If, instead, our knowledge of the true state of the system is limited such that we can only describe it as an epistemic state $\bm{\mu}^{(\mathrm{P})}$, the probability of the same event $\e_k$ conditioned on this preparation of the system is given by,
\begin{equation}
    p(\e_k|\pp)=\langle \bm{\xi}^{k}, \bm{\mu}^{(\mathrm{P})}\rangle.
\end{equation}
The stochastic property of the epistemic state and the unit response function guarantee that $\langle \bm{\mu}^{(\mathrm{P})},\mathbf{1}\rangle=1$.

An ontological model of a COPE $\mathbf{C}^a$ composed of $l$ column-stochastic submatrices as in Eq.~\eqref{eq: C^a} is given by an NMF of the form $\smash{\mathbf{C}^a=\mathbf{RE}}$, where the factors $\mathbf{R}$ and $\mathbf{E}$ are entry-wise nonnegative.
If $\mathbf{R}$ and $\mathbf{E}$ are column sum-$l$ and column-stochastic, respectively, the columns of $\mathbf{E}$ would have the form of epistemic states, and the rows of $\mathbf{R}$ sum to $l\mathbf{1}$.
It can always be assumed that these conditions are met because, if not, they can be achieved using a diagonal scaling as
\begin{equation}
            \mathbf{C}^{a}=\mathbf{R}'\mathbf{E}'
    =(\mathbf{R}'\mathbf{D}^{-1}) (\mathbf{D}\mathbf{E}')
    =\mathbf{R}\mathbf{E},
\end{equation}
where $\mathbf{D}^{-1}$ is a diagonal matrix that rescales the columns of $\mathbf{R}'$, and we assume $\mathbf{C}^{b}$ has no zero rows or columns.
Since $\mathbf{C}^{a}$ and $\mathbf{R}$ are column sum-$l$, $\mathbf{E}$ is column-stochastic.
The nonnegative rank (NNR) of $\mathbf{C}^{a}$, i.e., the smallest inner dimension $k$ at which an NMF is possible, determines the size of the smallest ontological model of $\mathbf{C}^{a}$.
Notice that by simply rescaling $\mathbf{R}$, we also get an ontological model of $\mathbf{C}^{a}$ with the same ontic size, which is why both forms of the COPE matrix can be used to describe the same system.

\subsection{Noncontextual Ontological Model}
The statistics of any physical system can be reproduced by an ontological model. However, a natural criterion for a system to be considered \textit{classically} explainable is that such a model is \textit{noncontextual}. In this work, we adopt the concept of generalized \textit{noncontextuality}, as introduced in Ref.~\cite{Spekkens_2005}.
This framework relies on the operational equivalence of experimental procedures to define contextuality.

The operational equivalence $(\simeq)$ of preparations (measurement events) means that there are no measurement events (preparations) in the operational theory that can distinguish them~\cite{Spekkens_2005}, i.e.,
\begin{equation}
\begin{split}
& \pp_i \simeq \pp_j \; \text{iff} \; p(\e_k| \pp_i, \MM) = p(\e_k|\pp_j, \MM) \; \forall \e_k, \MM, \\
& \e_k \simeq \e_l \; \text{iff} \; p(\e_k| \pp, \MM) = p(\e_l|\pp, \MM) \; \forall \pp.
\end{split}
\end{equation}
Motivated by Leibniz's principle of the (ontological) identity of empirical indiscernibles~\cite{spekkens2019ontologicalidentityempiricalindiscernibles}, we may require any model of the data to satisfy the \textit{broad noncontextuality hypothesis}~\cite{Shahandeh2021}.
The latter suggests that our models of physical phenomena should rely on equivalence classes instead of individual contexts.
Then, we can define the preparation (measurement) contextuality of an ontological model as the property that there are at least two different preparations (measurement events) that are operationally equivalent but have different ontic representations.
Conversely, in a noncontextual model, operationally equivalent procedures are represented by the same function over the ontic space, i.e.,
\begin{equation}\label{eq:ctxt_def}
\begin{split}
    & \pp_1\cong\pp_2 \Leftrightarrow \bm{\mu}^{(\pp_1)} = \bm{\mu}^{(\pp_2)},\\
    & \e_1 \cong \e_2 \Leftrightarrow \bm{\xi}^{(\e_1)} = \bm{\xi}^{(\e_2)},    
\end{split}
\end{equation}
where $\bm{\mu}^{(\pp_i)}$ is the epistemic state of the system prepared by $\pp_i$, and $\bm{\xi}^{(\e_i)}$ is the response function corresponding to the event $\e_i$ ($i=1,2$).

Let us show how these definitions translate to the COPE picture via an example.
Consider the COPE matrix of a GPT known as the \textit{box world},
\begin{equation}
    \begin{split}
    \mathbf{C}&=\begin{pmatrix}
        1 & 1 & 0 & 0\\
        0 & 0 & 1 & 1\\
        1 & 0 & 1 & 0\\
        0 & 1 & 0 & 1
    \end{pmatrix}.
    \end{split}
\end{equation}
It can easily be checked that an NMF of this matrix with inner dimension three does not exist.
Hence, it does not admit an ontological model of dimension three or less.
The smallest ontological model possible is thus given as
\begin{equation}
\begin{split}
    \mathbf{C}=
    \begin{pmatrix}
        1 & 1 & 0 & 0\\
        0 & 0 & 1 & 1\\
        1 & 0 & 1 & 0\\
        0 & 1 & 0 & 1
    \end{pmatrix}
    \begin{pmatrix}
        1 & 0 & 0 & 0\\
        0 & 1 & 0 & 0\\
        0 & 0 & 1 & 0\\
        0 & 0 & 0 & 1
    \end{pmatrix}
    =\mathbf{RE},
\end{split}
\end{equation}
which shows a trivial NMF as $\mathbf{E}=\mathds{1}$.
However, this model is contextual model because there are two operationally indistinguishable preparations, which are represented with two different vectors in the ontological model. Specifically, the average of the first and last column of the COPE matrix and the average of its second and third column are equal.
However, the average of the first and last column of the epistemic state matrix $\mathbf{E}$ is not equal to the average of its second and third column.
Therefore, at least two epistemic states exist in the model that the response functions cannot distinguish.
The latter is formally reflected in the rank separation of the matrices, that is, $\rank\mathbf{R}\neq\rank\mathbf{E}$.

As shown in the example above, the representations of the response functions (epistemic states) in a model might not span the full space of ontic states.
Hence, in a contextual model, there can be two epistemic states (response functions) whose difference is orthogonal to the span of the response functions (epistemic states), so that they cannot be distinguished.
This leads us to the following lemma and corollary from Ref.~\cite{ourpaper}.
\begin{lemma}
    An ontological model is preparation (measurement) noncontextual if and only if its set of response functions (epistemic states) separates the points in the linear span of its epistemic states (response functions).
\end{lemma}
\begin{corollary}
    A noncontextual operational theory with a COPE $\mathbf{C}$ admits an NMF $\mathbf{C}=\mathbf{RE}$ such that
\begin{equation}\label{eq:nctxt_rank}
   \rank{\mathbf{C}}=\rank{\mathbf{R}}=\rank{\mathbf{E}}.
\end{equation}
\end{corollary}
We call such a factorization
a \textit{equirank nonnegative matrix factorization} (ENMF).
Furthermore, we define the \textit{equirank nonnegative rank} (ENNR) of $\mathbf{C}$ as the smallest inner dimension where an ENMF is possible, which is equal to the smallest possible ontic size of a noncontextual model for $\mathbf{C}$.

The advantage of recasting noncontextuality in terms of matrix factorizations is that it does not require operational equivalences and operational identities to be explicitly stated.
It also allows us to marry computational problems concerning contextuality with problems and results from linear algebra and computational complexity.
For example, we obtain our first complexity result as follows.
\begin{corollary} \label{cor:lower bound}
    Assuming the Exponential Time Hypothesis, the complexity of computing the size of the smallest noncontextual model of an $n\times m$ COPE of rank $r$ is at least $(nm)^{O(r)}$.
\end{corollary}
\begin{proof}
    By computing the size of the smallest noncontextual model, by extension, it is decided whether an ENMF of inner dimension $r$ is possible. This is equivalent to deciding if an NMF of inner dimension $r$ is possible, because any NMF of inner dimension $r$ is also an ENMF. 
    The complexity of the NMF decision is known to be at least $(nm)^{O(r)}$~\cite{arora2011computingnonnegativematrixfactorization}, assuming the Exponential Time Hypothesis~\cite{impagliazzo2001complexity}.
\end{proof}

In this paper, we are concerned with a geometric interpretation of ENMFs and the complexity of determining the ENNR of operational theories.

\subsection{Generalized Probabilistic Theories}
A GPT is an alternative model that can be used to explain the observed outcome statistics of experiments.
It consists of a set of states $\mathcal{S}:=\{\bm{s}^{i}\}_i$ which lie in a real vector space $\mathfrak{V}$ and a set of effects $\mathcal{E}:=\{\bm{e}^{i}, \bm{u}\}_i$, with $\bm{u}$ being the unit effect, lying in the dual vector space $\mathfrak{V}^*$.
Here, the dual of a polytope is defined in the following way. Given a polytope $\mathcal{K}$ in a real inner product space $(\mathfrak{V},\langle\cdot,\cdot\rangle_{\mathfrak{V}})$, its dual $\mathcal{K}^*$ is given by,
\begin{equation}
    \mathcal{K}^*:=\{\bm{x}\in \mathfrak{V}|\langle \bm{x},\bm{s}\rangle \in [0,1] \forall \bm{s}\in\mathcal{K}\}.
\end{equation}
In finite dimensions, thanks to the Riesz theorem, there is an isomorphism between $\mathfrak{V}^*$ and $\mathfrak{V}$, so we can assume that the set of effects is also in $\mathfrak{V}$.

Within the GPT framework, a state is a representation of a class of operationally equivalent preparation procedures, and an effect is a representation of a class of operationally equivalent events.
The effects in a measurement sum to the unit effect $\bm{u}$.

Let $\bm{e}^{j}$ be the effect corresponding to the class $[\e_j]$ of events operationally equivalent to the $j$th outcome of a particular measurement $\MM$. Then the probability of observing this event given that the system was prepared using $[\pp_i]$, i.e., any preparation procedure operationally equivalent to $\pp_i$, is given by,
\begin{equation}
    p([\e_j]|[\pp_i],[\MM])=\langle \bm{e}^{j},\bm{s}^{i}\rangle.
    %p(\e_j|\bm{s}^{i})=\langle \bm{s}^{i},\bm{e}^{j}\rangle.
\end{equation}
Here, $\langle\cdot , \cdot \rangle$ refers to the inner product over $\mathfrak{V}$.
Since the event probabilities of a particular measurement must sum to 1, the states should be normalized as,
\begin{equation}\label{unit}
    \langle \bm{s}^{i},\bm{u}\rangle=1.
\end{equation}

A GPT model of a COPE $\mathbf{C}^a$ can be found by simply finding a real factorization of the form,
\begin{align}\label{eq:C_mindec}
    \mathbf{C}^a=AB,
\end{align}
where $A \in \mathbb{R}^{m\times k}$, $B \in \mathbb{R}^{k\times n}$ is column stochastic, and they are such that $\rank\mathbf{C}^a=\rank A=\rank B$.
Note that, because $\mathbf{C}^a$ has no zero columns, we can rescale any decomposition using a diagonal matrix $D$ as $\mathbf{C}^a=A'B'=(A'D^{-1}) (DB')=AB$ so that $B$ is column-stochastic.
Also, note that since $\mathbf{C}^a$ represents an experiment with $l$ measurements, it is divided into submatrices $C^{ai}$, each corresponding to a particular measurement, as illustrated in Eq.~\eqref{eq: C^a}. 
These submatrices are each column stochastic, therefore, the sum of rows of the corresponding submatrices of $A$, $A^{i}$, are equal.

Now, we have a GPT for the COPE where each row of $A$ is an effect vector and the rows within a subset of rows corresponding to a particular measurement sum to the unit effect $\bm{u}$, and each column of $B$ is a normalized state vector. 
Such a factorization of $\mathbf{C}$ gives the simplest (lowest-dimensional) GPT that can reproduce it whenever the inner dimension $k$ of the decomposition equals the rank of $\mathbf{C}$.
The equality between ranks guarantees that every equivalence class is represented by a unique vector in the GPT as required.

\section{Geometric Interpretation of ENMF} \label{section 3}

It is easy to visualize any nonnegative $m \times n$ matrix geometrically: Simply interpret the columns as $n$ vectors in the positive orthant of an $m$-dimensional Euclidean space, corresponding to the vertices of some geometrical body. In particular, if the columns are convexly independent, the vertices will form the extreme points of a bounded convex polytope.
Building on this simple intuition, one can also devise different geometrical interpretations of the decompositions of such matrices, including NMF and ENMF.
In this section, we provide two such interpretations for ENMF.
This leads to an algorithm for deciding whether a COPE admits a noncontextual model, as well as an algorithm for finding its smallest noncontextual model.

To do this, we first consider how the matrix factors of the COPE can be represented as polyopes.
\begin{theorem} \label{theor:restrictedP1}
Let $\mathbf{C}$ be an $m\times n$ COPE matrix of rank $r$, which has $l$ measurements.
Let $\mathbf{C} =AB$ be a real matrix factorization of $\mathbf{C}$ with inner dimension $r$.
Let $\mathcal{A}$ be the polyhedron $\mathcal{A}:=\{\bm{x}\in\mathbb{R}^r|A\bm{x}\geq0, \mathbf{1}^\top A\bm{x}=l\}$, $\mathcal{B}$ the convex hull of the columns of $B$, and $\mathcal{B}^*:=\{\bm{x}\in\mathbb{R}^r|\bm{x}^\top B\leq 1\}$ the dual polytope of $\mathcal{B}$.
Then, $\mathcal{B}\subseteq\mathcal{A}\subseteq \mathcal{B}^*$.
Furthermore, $\mathcal{A}$ is bounded.

\end{theorem}
\begin{proof}
    Since $AB=\mathbf{C}\geq0$ and $\mathbf{C}$ is divided into $l$ column-stochastic subsets of rows, each column of $B$ satisfies $AB_{:i}\geq0$ and $\mathbf{1}^\top AB_{:i}=l$, therefore it is contained in $\mathcal{A}$. Furthermore, for any such subset $\mathbf{C}^i$ of rows of $\mathbf{C}$, the rows within corresponding subset of rows $A^i$ of $A$ sum to the unit effect $\bm{u}$. Now, we note that for any $\bm{x}\in\mathcal{A}$, $\mathbf{1}^\top A\bm{x}=l$ and $\mathbf{1}^\top A=l\bm{u}^\top$, since there are $l$ measurements and the sum of effects within each measurement is $\bm{u}$. It follows that $\bm{u}^\top \bm{x}=1$.
    Therefore, given that $A\bm{x}\geq0$, we have $0\leq A\bm{x}\leq1$.
    This implies that $\mathcal{A}$ is contained in $B^*$.
\end{proof}
We observe that $\mathbf{C}_{ij}$ is the slack of vertices of $\mathcal{B}$ with respect to the hyperplanes of $\mathcal{A}$. Also, $\mathcal{A}$ is the set of all logically possible normalized states given the effects.

Given the nonnegative matrix $\mathbf{C}$, the following two problems are also shown to be equivalent~\cite{vavasis}:
\begin{enumerate}
    \item Decide if the nonnegative rank and the rank of $\mathbf{C}$ are equal.
    \item Decide if a nested simplex $\mathcal{G}$ exists such that $\mathcal{B}\subseteq\mathcal{G}\subseteq\mathcal{A}$. Here, the polytopes $\mathcal{A}$ and $\mathcal{B}$ are defined via a decomposition of the form given in Theorem~\ref{theor:restrictedP1}.
\end{enumerate}
We generalize this result to map the problem of deciding whether an ENMF of inner dimension $k$ of $\mathbf{C}$ exists to a nested simplex problem.
\begin{lemma} \label{lemma:lemma4}
Let $\mathbf{C}$ be a nonnegative $m \times n$ matrix of rank $r$ whose columns each sum to $l$, and suppose that $\mathbf{C} = AB$ is a real matrix factorization of $\mathbf{C}$ with inner dimension $r$.
Define the polytopes $\mathcal{A},\mathcal{B}\subset\mathbb{R}^{r}$ from $B$ and $A$ as in Theorem~\ref{theor:restrictedP1} so that $\mathcal{B}\subseteq\mathcal{A}$.
There exists an ENMF of $\mathbf{C}$ of inner dimension $k$ if and only if there exists a polytope $\mathcal{B}\subseteq\mathcal{G}\subseteq\mathcal{A}$ with $k$ vertices, and a rank-$r$ nonnegative matrix $\mathbf{E}$ of convex coefficients that maps the matrix of vertices of $\mathcal{G}$ to the matrix of vertices of $\mathcal{B}$.
\end{lemma}

\begin{proof}
We first show the only if direction. 
When an ENMF $\mathbf{C}=\mathbf{R}\mathbf{E}$ of inner dimension $k$ exists, there exist rank-$r$ matrices $G \in \mathbb{R}^{r \times k}$ and $F \in \mathbb{R}^{k \times r}$ such that $\mathbf{R}=AF$ and $\mathbf{E}=GB$.
Without loss of generality, we assume that the columns of $\mathbf{R}$ each sum to $l$, and that $\mathbf{E}$ is column-stochastic.
We can assume that $G$ and $F$ have no zero rows and columns, respectively.
This is because if the $j$th row of $F$ is zero, then
the $j$th row of $\mathbf{E}=FB$ will be zero.
The $j$th column of $G$ can also be set to zero, making the $j$th column of $\mathbf{R}=AG$ zero without changing the product $AGFB$.
This means, in turn, that we can remove the $j$th row of $\mathbf{E}$ and the $j$th column of $\mathbf{R}$ to obtain an ENMF in $k-1$ dimensions.

Since $\mathbf{R}=AG$ is nonnegative and has a column sum of $l$, the columns of $G$ are in $\mathcal{A}$, therefore, we construct the polytope $\mathcal{G}$ as the convex hull of the columns of $G$. Given that $A$ and $B$ are full-rank, $AGFB=AB$ implies that $GFB=G\mathbf{E}=B$. Since $\mathbf{E}$ is column-stochastic, the columns of $B$ are given by convex combinations of the columns of $G$. Equivalently, $\mathcal{B}\subseteq\mathcal{G}$.

Furthermore, since $\mathbf{R}$ and $\mathbf{E}$ have factors of $A$ and $B$, respectively, $\rank{(\mathbf{R})},\rank{(\mathbf{E})}\leq r$.
Then, $\mathbf{R}\mathbf{E} = C$ implies that $\operatorname{rank}(\mathbf{R})\operatorname{rank}(\mathbf{E}) \geq \operatorname{rank}(\mathbf{C}) = r$,
therefore $\operatorname{rank}(\mathbf{R}) =\operatorname{rank}(\mathbf{E})= r$. 

For the reverse direction, suppose there exists a nested polytope $\mathcal{G}$ satisfying $\mathcal{B}\subseteq \mathcal{G} \subseteq \mathcal{A}$, of $k$ vertices, and a rank-$r$ column-stochastic matrix $\mathbf{E}$ that maps the vertices of $\mathcal{G}$ to the vertices of $\mathcal{B}$. 
Given $\mathcal{G}$ and $\mathcal{B}$ have associated matrices $G$ and $B$, respectively, we have $G\mathbf{E} = B$ (up to a permutation of the columns of $B$, which can be arbitrarily changed by permuting the rows of $\mathbf{E}$). Defining $B^{\dagger}$ as the pseudoinverse of $B$, we can define $F := \mathbf{E}B^{\dagger}$.
\end{proof}

Lemma~\ref{lemma:lemma4} boils the complexity of ENMF problem down to two parts. First, the complexity of finding the polytope $\mathcal{G}$.
Second, the complexity of deciding whether a convex coefficient matrix $\mathbf{R}$ of rank $r$ exists and its computation.
Note that, while a nonnegative transformation from $G$ to $A$ can be found using methods such as nonnegative least squares, such generic approaches do not provide rank guarantees.
We thus begin by showing that our desired transformation can be decided and found in polynomial time.

\begin{lemma} \label{lemma:algorithm}
Let $G\in\mathbb{R}^{k\times r}$ and $B \in \mathbb{R}^{m \times r}$ be rank-$r$ matrices corresponding to polytopes $\mathcal{G}$ and $\mathcal{B}$, respectively, such that $\mathcal{B}\subseteq\mathcal{G}$.
For a bit length $b$ for specifying the entries of $G$ and $B$,
it can be decided in poly($r,k,m,b$) time whether there exists 
a rank-$r$ nonnegative matrix $\mathbf{E}$ such that $G\mathbf{E}=B$.
When yes, $\mathbf{E}$ can be found within an element-wise $\delta$ error in poly($r,k,m,l,\log(\delta^{-1})$) time.
\end{lemma}

\begin{proof}
Let $G^{\dagger}$ be the pseudoinverse of $G$. Let $\bar{E}:=G^{\dagger}B$. If $\bar{E}$ is nonnegative then $\mathbf{E}=\bar{E}$. If not, we search for a shear transformation $L$ of $\bar{E}$ such that $L\bar{E}\geq0$ and $G\bar{E}=GL\bar{E}$. 
The shear matrix $L$  must be of the form,
\begin{align} \label{shear}
    L=\mathds{1}+\sum_{i=1}^r\sum_{j=1}^{k-r}D_{ij} \bm{b}^{j}\bm{a}^{i\top},
\end{align}
wherein $\{\bm{a}^{1\top},...,\bm{a}^{r\top}\}$ is an orthonormal basis for the rows of $G$, $\{\bm{b}^{1},...,\bm{b}^{(k-r)}\}$ is an orthonormal basis for the right kernel of $G$, and
$D$ is a real matrix to be found via optimization. 

Define the objective function $f(D)$ as the modulus of the sum of negative entries in $L\bar{E}$.
Minimizing $f(D)$ leads to a solution when $f(D)=0$.
In Appendix~\ref{app:L_optim}, we give the details of a linear program for optimizing $D$.
We show that the number of constraints and search space size are polynomial in $m$, $r$ and $k$.
The respective complexities of deciding if a solution to a linear program exists, and computing it when it does, are well-known to be poly($r,k,m,l$) and poly($r,k,m,l,\log(\delta^{-1})$), respecively~\cite{ASPVALL19801, Renegar1988}.
\end{proof}

\begin{corollary} \label{col:detecting}
Given $m\times n$ nonnegative matrices $\mathbf{C}$ of a fixed rank $r$,
it can be decided in poly($m,n,l$) time whether an ENMF of unspecified inner dimension exists.
Here, $b$ is the bit length for specifying the entries of $\mathbf{C}$.
\end{corollary}

\begin{proof}
Firstly, find a rank factorization of $\mathbf{C}$ and the corresponding inner and outer polytopes $\mathcal{B}$ and $\mathcal{A}$ as in Lemma~\ref{lemma:lemma4}.
The number of vertices of $\mathcal{A}$, which has $m$ facets, as well as the complexity of finding these vertices from the facets, is polynomial in $m$.
This polynomial is of degree $r-1$ as shown in Sec.~5 of Ref.~\cite{vertexenumeration}.
Note that, $\mathcal{B}\subseteq\mathcal{G}\subseteq\mathcal{A}$.
Hence, $\mathcal{G}=\mathcal{A}$ is a possible nested polytope. 
Then, using Lemma~\ref{lemma:algorithm}, where we choose $k$ to be the number of vertices of $\mathcal{A}$, which is polynomial in $n$ for fixed $r$,
a rank-$r$ nonnegative matrix $\mathbf{E}$ that maps the vertices of $\mathcal{A}$ to the vertices of $\mathcal{B}$ can be decided and computed in polynomial time.
Since $\mathcal{A}$ is the largest possible nested polytope, such a matrix exists if and only if an ENMF exists.    
\end{proof}

Corollary~\ref{col:detecting} might tempt us to think that deciding whether an ENMF, and thus a noncontextual ontological model, of a given COPE matrix (without specifying the inner dimension $k$) exists is efficient.
For one thing, according to the proof, we can choose the trivial nested polytope $\mathcal{G}=\mathcal{A}$ and check the existence of $\mathbf{E}$ in polynomial time.
However, we emphasize that this observation is only valid if the rank of the input COPE matrices is constant.
We believe that the complexity of this simplest check grows exponentially with the rank.
The reason is that checking the existence of $\mathbf{E}$ requires specifying \textit{vertices} of $\mathcal{A}$ (known as the \textit{V-representation}), while the description of $\mathcal{A}$ is in terms of its hyperplanes (known as the \textit{H-representation}).
The problem of finding the V-representation of a polytope from its H-representation is called the \textit{vertex enumeration problem}.
While this is known to be polynomial in special cases, e.g. for fixed dimension or if the polytope is simple or simplicial, the problem is believed to be exponential on general instances because the number of vertices grows exponentially with the number of defining hyperplanes. 
In particular, the problem is known to be NP-hard for unbounded polydera~\cite{Khachiyan2008}.

The problem of finding a polytope $\mathcal{G}$ with $k$ vertices nested between two polytopes $\mathcal{B}$ and $\mathcal{A}$ is well-known in computational geometry as the \textit{nested polytope problem} (NPP).
A natural approach thus would be to use algorithms solving the general NPP to find $\mathcal{G}$.
Unfortunately, this is not possible for two reasons. 
First, it is known that NPP is NP-hard, and in fact, $\exists\mathbb{R}$-complete, in dimensions greater than two~\cite{Das1990,dobbins2019}.
Second, when a solution for $\mathcal{G}$ with a specified $k$ exists, it is not guaranteed that the solution's vertices can be mapped to $\mathcal{B}$ by a rank-$r$ nonnegative matrix $\mathbf{E}$. This extra restriction calls for devising a task-specific algorithm.
In what follows, we will provide such an algorithm through a geometric interpretation.

\begin{lemma} \label{lemma:cone}
Let $\mathbf{C}$ be a nonnegative $m \times n$ matrix of rank $r$ whose columns each sum to $l$ and $\mathbf{C}=AB$ a real matrix factorization of $\mathbf{C}$ with inner dimension $r$.
For a given integer $k$, construct $\bar{A} \in \mathbb{R}^{m\times k}$
and $\bar{B} \in \mathbb{R}^{k\times n}$ by appending the $m\times (k-r)$ and $(k-r)\times n$ zero matrices to the left of $A$ and top of $B$, respectively. 
Define the unbounded polyhedral cone $\bar{\mathcal{A}}$ as the set $\{\bm{x}\in\mathbb{R}^k|A\bm{x}\geq0, \mathbf{1}^\top \bar{A}\bm{x}=l\}$.
$\bar{\mathcal{B}}$ is the corresponding polytope of $\bar{B}$.
There is a solution to ENMF of $\mathbf{C}$ of inner dimension $k$ if and only if there is a ($k-1$)-simplex $\mathcal{\bar{G}}$ such that $\bar{\mathcal{B}}\subseteq \mathcal{\bar{G}}\subseteq \bar{\mathcal{A}}$.
\end{lemma}

\begin{proof}
Let $\mathbf{C}=(AG)(FB)$ be a solution to the ENMF of $\mathbf{C}$ of inner dimension $k$ as in Lemma~\ref{lemma:lemma4}, where $\mathbf{R}:=AG$ and $\mathbf{E}=FB$ are column sum-$l$ and column-stochastic, respectively.
Note that $G$ is not square in general; hence, it does not represent a simplex, which is a full-dimensional bounded convex body. 
Thus, we aim to show that there exists a full-rank square matrix $\bar{G}'$ such that $\mathbf{C}=\bar{A}\bar{G}' \bar{G}'^{-1}\bar{B}$.
Then, since $\bar{G}'^{-1}$ is full-rank ($k$) and $\bar{B}$ is rank-$r$, we have $\mathbf{E}=\bar{G}'^{-1}\bar{B}$ which is rank-$r$ and nonnegative.

We begin by observing that due to the zeros in $\bar{A}$ and $\bar{B}$, we can append two $(k-r)\times k$ and $k\times(k-r)$ random real matrices to the top of $G$ and left of $F$, respectively, so that $\bar{A}\bar{G}\bar{F}\bar{B}=\mathbf{C}$.
Clearly, we also have $\bar{A}\bar{G}\geq0$ and $\bar{F}\bar{B}\geq0$.

Furthermore,
$\bar{G}\bar{F}\bar{B}$ can only differ from $\bar{B}$ in its first $(k-r)$ rows because
$\bar{A}\bar{G}\bar{F}\bar{B}=\bar{A}\bar{B}$, the first $(k-r)$ rows of $\bar{B}$ are all zeroes, and the last $r$ columns and rows of $\bar{A}$ and $\bar{B}$, respectively, give a rank factorization of $\mathbf{C}$.
Therefore, there exists a full-rank shear matrix $L$ of the form,
\begin{equation}
\begin{split}
    L &= \mathds{1} + \sum_{i=1}^{k-r} \sum_{j=k-r+1}^{k} (D_{ij} \bm{a}^{i} (\bm{a}^{j})^\top),
\end{split}
\end{equation}
where $\bm{a}^{i}$s are the columns of $\mathds{1}_{k\times k}$ and $D$ is a real matrix, such that $\bar{G}\bar{F}\bar{B}=L\bar{B}$. The inverse of $L$ is given by,
\begin{equation}
        L^{-1} = \mathds{1} - \sum_{i=1}^{k-r} \sum_{j=k-r+1}^{k} (D_{ij} \bm{a}^{i} (\bm{a}^{j})^\top) \label{sinverse}.
\end{equation}
We have,
\begin{equation}
\begin{split}
    \mathbf{C} &= \bar{A}\bar{G}\bar{G}^{-1}\bar{G}\bar{F}\bar{B},\\
    &= \bar{A}\bar{G}\bar{G}^{-1}L\bar{B}.
\end{split}
\end{equation}
Using Eq.~\eqref{sinverse}, $\bar{A}\bar{G}=\bar{A}L^{-1}\bar{G}$ because
$L^{-1}\bar{G}$ only differs from $\bar{G}$ in the first $(k-r)$ rows which leads to,
\begin{equation}
\begin{split}
    \mathbf{C}&=\bar{A}L^{-1}\bar{G}\bar{G}^{-1}L\bar{B}\\
    &=(A\bar{G}^{\prime})(\bar{G}^{\prime-1} \bar{B}),\\
    \mathbf{R}&=\bar{A}\bar{G}=A\bar{G}^{\prime},\\
    \mathbf{E}&=\bar{F}\bar{B}=\bar{G}^{\prime-1}\bar{B}.
\end{split}
\end{equation}
Now, $\bar{G}^{\prime}$ is square and full-rank.
Hence, its corresponding polytope $\bar{\mathcal{G}}^{\prime}$ is a $(k-1)$-simplex.
Furthermore, $\mathbf{R}=\bar{A}\bar{G}^{\prime}$ and $\mathbf{E}=\bar{G}^{\prime-1}\bar{B}$ being nonnegative, and $\mathbf{E}$ being also column-stochastic imply that $\bar{\mathcal{G}}^{\prime}\subseteq \bar{\mathcal{A}}$ and $\mathcal{\bar{B}}\subseteq\mathcal{G}'$, respectively. Hence the claim.

For the reverse direction, assume there exists a full-dimensional nested $(k-1)$-simplex $\mathcal{G}$ with a corresponding full-rank matrix $G$.
Then $\mathcal{G}\subseteq\bar{\mathcal{A}}$ implies that $R:=\bar{A}G$ is nonnegative and has a column sum of $l$.
$\bar{\mathcal{B}}\subseteq \mathcal{G}$ implies that $\mathbf{E}:=\bar{A}G^{-1}$ is a matrix of convex coefficients mapping $\mathcal{G}$ to $\mathcal{A}$, so it is nonnegative and column-stochastic.
Finally, $\mathbf{E}$ is unique and $\rank\mathbf{E}=\rank\bar{B}=r$ because $G$ and $G^{-1}$ are full-rank.
\end{proof}
The advantage of the intermediate simplex approach of Lemma~\ref{lemma:cone} is that, once found, it is guaranteed that the rank of $\mathbf{E}$ is $r$.
As a result, following finding the intermediate simplex, one can use other efficient algorithms such as nonnegative least squares to compute $\mathbf{E}$.
The main question is thus, what is the complexity of the \textit{intermediate simplex problem}?
This was shown by Vavasis~\cite{vavasis} to be NP-hard, assuming that the inner polytope is full-dimensional and the outer polytope is bounded.
In our construction in Lemma~\ref{lemma:cone}, however, neither of these conditions are met.
Our next task is thus to transform this construction into one wherein $\bar{\mathcal{B}}$ is full-dimensional and $\bar{\mathcal{A}}$ is bounded.
However, we first show that the intermediate simplex $\mathcal{G}$ of Lemma~\ref{lemma:cone} can always be assumed to have a specific orientation in the subspace orthogonal to the affine span of the inner polytope $\bar{\mathcal{B}}$.

\begin{lemma} \label{orthogonal transformation}
    For a rank-$r$ nonnegative matrix $\mathbf{C}$,
    let  $\bar{\mathcal{B}}\subseteq\bar{\mathcal{G}}\subseteq\bar{\mathcal{A}}$ be nested polytopes as in Lemma~\ref{lemma:cone}.
    Then, transforming $\bar{\mathcal{G}}$ by any orthogonal transformation acting on the subspace orthogonal to ${\rm span}(\bar{\mathcal{B}})$ also results in a nested simplex.
\end{lemma}
\begin{proof}
    Such a transformation corresponds to multiplying the matrix $G$ on the right by a matrix of the form,
    \begin{equation}
        O=\begin{pmatrix}
        O_{(k-r)\times(k-r)} & \mathbf{0}\\
        \mathbf{0} & \mathds{1}
    \end{pmatrix},
    \end{equation}
    where $O_{(k-r)\times(k-r)}$ is an orthogonal matrix.
    Note that, for any matrix $\mathbf{Q}$ such that the top $k-r$ rows of $G\mathbf{Q}$ are zeros, the top $k-r$ rows of $GO\mathbf{Q}$ will also be zeros.
    Then, since $O$ acts as the identity on the subspace $\operatorname{span}(\bar{\mathcal{B}})$, $\mathbf{Q}G=\bar{\mathcal{B}}$ would imply $\mathbf{Q}OG=\bar{\mathcal{B}}$.
\end{proof}

\begin{theorem} \label{thm:reduction}
For a rank-$r$ nonnegative matrix $\mathbf{C}$ with column sum $l$, let $\bar{A}$ and $\bar{B}$ be matrices as in Lemma~\ref{lemma:cone}.
There exists an $(m+k-r)\times(n+2k-2r)$ rank-$k$ nonnegative matrix $\bar{\mathbf{C}}$ for $\smash{k>r}$ such that a
ENMF of $\mathbf{C}$ of inner dimension $k$ exists if and only if the NNR of $\bar{\mathbf{C}}$ equals $k$.
\end{theorem}
\begin{proof}
Recall that $\bar{\mathcal{A}}$ and $\bar{\mathcal{B}}$ are the trivial embeddings of the $(r-1)$-dimensional polytopes $\mathcal{A}$ and $\mathcal{B}$ in $k-1$ dimensions. 
Our first goal is to bound $\bar{\mathcal{A}}$ by adding extra hyperplanes. 
This can be done as follows:
For each extra dimension added to $\mathcal{A}$, extending it to an unbounded prism, we need two hyperplanes to bound the convex body in that direction.
To do this, define $\bm{h}^i$ as $\bm{h}^i_j=\delta_{ij}$ for $1\leq j \leq k$. Then, add $2(k-r)$ rows of length $k$ to the bottom of $\bar{A}$, of the form $-\bm{h}^i/2+\bm{u}/2$ for $1\leq i \leq k-r$ and $\bm{h}^i/2+\bm{u}/2$ for $k-r+1\leq i \leq 2(k-r)$. This corresponds to adding $k-r$ extra dichotomic measurements, meaning $l\rightarrow l+k-r$, but $\bm{u}$ does not change.
Note that each $\bm{h}^i$ is orthogonal to the span of $\bar{\mathcal{A}}$.
Let us denote by $\bar{A}_b$ and $\bar{\mathcal{A}}_b$ the resulting $k \times (n+2k-2r)$ matrix and its corresponding $(k-1)$-dimensional bounded polytope, respectively.

Assume that there exists a $(k-1)$-simplex $\mathcal{G}$ nested between $\bar{\mathcal{B}}$ and $\bar{\mathcal{A}}$.
We want the description of $\mathcal{G}$ to be as simple as possible. Therefore, using Lemma.~\ref{orthogonal transformation}, we assume that the vertex of $\mathcal{G}$ furthest from the span of $\bar{\mathcal{B}}$ has no $\bm{h}^j$ component for $j>1$. Further, we assume that the vertex of $\mathcal{G}$ furthest from $\operatorname{span}(\bar{\mathcal{B}}\cup_{j=1}^i\bm{h}^j)$ has no $\bm{h}^p$ component for $p>i+1$ and $i=1,...,r-k$.

To obtain a nested simplex between $\bar{\mathcal{B}}$ and $\bar{\mathcal{A}}_b$, the vertices of $\mathcal{G}$ can be rescaled along the $\bm{h}^i$ directions to obtain the $(k-1)$-simplex $\bar{\mathcal{G}}$ such that $k-r$ of its vertices lie on the hyperplanes defined by $\bm{h}^i$s and $\bar{\mathcal{B}}\subseteq\bar{\mathcal{G}}\subseteq\bar{\mathcal{A}_b}$.

We now aim to extend $\bar{\mathcal{B}}$ to a full-dimensional polytope by adding $k-r$ vertices $\{\bm{v}_i\in \bar{\mathcal{G}}\}_{i=1}^{k-r}$ to it.
The first vertex can be constructed by adding a $\bm{h}^1$ component to a point in $\bar{\mathcal{B}}$.
Let this point be $\bm{c}^1$, the barycenter of $\bar{\mathcal{B}}$ computed from its vertices, so that $\bm{v}^1 = \bm{c}^1 + \bar{h}_1 \bm{h}^1$.
We thus need to estimate the maximum $\bar{h}_1$.

Suppose $\bar{\mathcal{B}}_{\bar{\mathcal{A}}}:=\operatorname{span}(\bar{\mathcal{B}})\cap\bar{\mathcal{A}}$ is the polytope formed by the intersection of the span of $\bar{\mathcal{B}}$ and $\bar{\mathcal{A}}$.
Similarly, define $\bar{\mathcal{B}}_{\bar{\mathcal{G}}}:=\operatorname{span}(\bar{\mathcal{B}})\cap\bar{\mathcal{G}}$ as the polytope formed by the intersection of the span of $\bar{\mathcal{B}}$ and $\bar{\mathcal{G}}$.
Note that $\bar{\mathcal{B}}\subseteq\bar{\mathcal{B}}_{\bar{\mathcal{G}}}\subseteq \bar{\mathcal{B}}_{\bar{\mathcal{A}}}$.

Consider the vertex $\bm{g}$ of $\bar{\mathcal{G}}$ on the hyperplane $\bm{x}^{\top}\bm{h}^1=1$ and denote its projection along $\bm{h}^1$ onto $\bar{\mathcal{B}}_{\bar{\mathcal{A}}}$ as $\bm{y}^1$.
The line segment $\ell$ connecting $\bm{y}^1$ to $\bm{c}^1$ intersects with $\bar{\mathcal{B}}_{\bar{\mathcal{G}}}$ at a point $\bm{z}$ opposite to $\bm{y}^1$.
The three points $\bm{g}$, $\bm{y}^1$, and $\bm{z}^1$ form a right triangle with its base $\overline{\bm{g}\bm{y}^1}$ having unit length.
Using the basic proportionality theorem of triangles, we have
\begin{equation}
    h_1=\frac{h_1}{\overline{\bm{g}\bm{y}^1}}=\frac{\overline{\bm{z}^1\bm{c}^1}}{\overline{\bm{z}^1\bm{y}^1}}.
\end{equation}
This is illustrated in Fig.~\ref{fig:height}.
\begin{figure}
    \centering
    \includegraphics[width=0.5\linewidth]{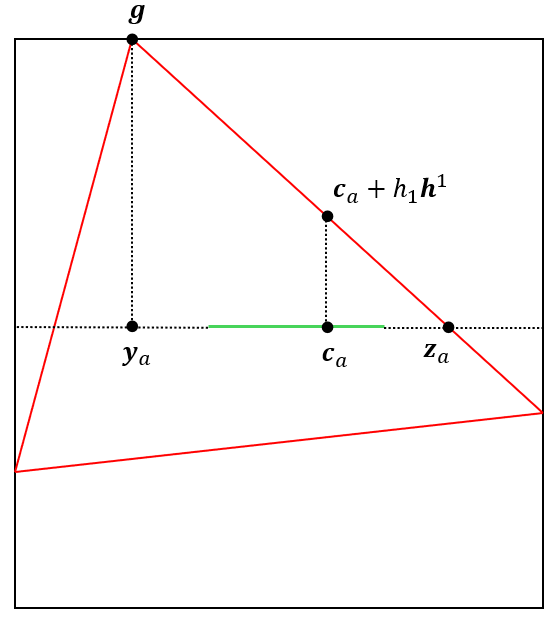}
    \caption{2D cross-section of the nested polytopes $\bar{\mathcal{A}}_b$ (solid black), $\bar{\mathcal{A}}$ (solid green), and $\bar{\mathcal{G}}$ (solid red).}
    \label{fig:height}
\end{figure}
Let $\mathcal{S}_{11}$ be a simplex whose vertices are in $\bar{\mathcal{B}}$ and contains $\bm{c}^1$. Since $\bm{c}^1$ is the barycenter of $\bar{\mathcal{B}}$, it must be the case that $\overline{\bm{z}^1\bm{c}^1}$ is greater than or equal to the closest point of $\partial\mathcal{S}_{11}$ to $\bm{c}^1$.
Let us call this distance $d_1$.
Furthermore, note that without loss of generality, we can assume $\bar{\mathcal{A}}$ contains the origin. Let $\mathcal{S}_{21}$ be a simplex that contains $\bar{\mathcal{B}}_{\bar{\mathcal{A}}}$. This can be given by the polar dual of a simplex that is contained in the polar dual of $\bar{\mathcal{A}}$ and contains the origin.
Since $\bm{y}^1, \bm{z}^1\in \bar{\mathcal{B}}_{\bar{\mathcal{A}}}$, $\overline{\bm{z}^1\bm{y}^1}$ is upper bounded by the maximum distance between any two vertices of $\mathcal{S}_{21}$.
We call this distance $d_2$.
It follows that,
\begin{equation}\label{eq: htilde}
    h_1=\frac{\overline{\bm{z}^1\bm{c}^1}}{\overline{\bm{z}^1\bm{y}^1}} \geq \frac{d_1}{d_2}:=\bar{h}_1.
\end{equation}
We iterate the above for $i=1,...,k-r$ using the following procedure. 
After the $(i-1)$th iteration, we so far have values for $\bar{h}_1,...,\bar{h}_{i-1}$. 
For the $i$th iteration, define $\bar{\mathcal{B}}_i=\operatorname{conv}(\bar{\mathcal{A}}\cup_{j=1}^{i-1}\bm{v}^j)$ where $\bm{v}^j = \bm{c}^j + \bar{h}_j\bm{h}^j$,
$\bar{\mathcal{B}}_{\bar{\mathcal{A}}i}=\operatorname{span}(\bar{\mathcal{B}}_i)\cap\bar{\mathcal{A}}$ and $\bar{\mathcal{B}}_{\bar{\mathcal{G}}i}=\operatorname{span}(\bar{\mathcal{B}}_i)\cap\bar{\mathcal{G}}$. 
Similarly to before, let $\bm{c}^i$ be the barycentre of $\bar{\mathcal{B}}_i$, let $\mathcal{S}_{1i}$ be a simplex contained in $\bar{\mathcal{B}}_i$ containing $\bm{c}^i$, and let $\mathcal{S}_{2i}$ be a simplex containing $\bar{\mathcal{B}}_{\bar{\mathcal{A}}i}$. Then, let $d_{1i}$ be the minimum distance from $\bm{c}^i$ to $\partial\mathcal{S}_{1i}$ and $d_{2i}$ be the maximum distance between two vertices of $\mathcal{S}_{2i}$. 
Then, using a similar line of reasoning as above, $h_i$ can be bounded as,
\begin{equation}\label{eq:h_i}
    h_i\geq \bar{h}_i= \frac{d_{1i}}{d_{2i}}.
\end{equation}
The iteration procedure above ensures that the distance between the boundary of $\bar{\mathcal{G}}$ and $c_i$ along the $\bm{h}^i$ direction is not infinitesimally small.

It follows that $\bar{\mathcal{B}}_i=\operatorname{conv}(\bar{\mathcal{A}}\cup_{k=1}^{k-r}\bm{v}^k)$ and
$\bar{\mathcal{B}}_b:={\rm conv}\{\bar{\mathcal{B}}\cup_{i=1}^{k-r}\bar{h}_i \bm{h}^i\}$ 
is a full dimensional bounded polytope such that $\bar{\mathcal{B}}_b\subseteq \bar{\mathcal{G}} \subseteq \bar{\mathcal{A}}_b$.
This polytope corresponds to an $(m+k-r)\times k$ rank-$k$ matrix $\bar{B}_b$, which has $k-r$ extra columns added to the right of those of $\bar{A}$.

Now $\bar{\mathcal{B}}_b$ and $\bar{\mathcal{A}}_b$ are both $(k-1)$-dimensional
polytopes and we define $\bar{\mathbf{C}}=\bar{A}_b\bar{B}_b$ as the $(m+2k-2r) \times (n+k-r)$ slack of vertices of $\bar{\mathcal{B}}_b$ with respect to the facets of $\bar{\mathcal{A}}_b$.
We note that $\bar{\mathbf{C}}$ contains $\mathbf{C}$ as a submatrix in its top-left corner.

Using Ref.~\cite{vavasis}, the existence of the intermediate simplex $\bar{\mathcal{G}}$ is equivalent to the existence of an NMF of inner dimension $k$ of $\bar{\mathbf{C}}$.
Since $\rank \bar{\mathbf{C}}=k$, it follows also that the NNR of $\bar{\mathbf{C}}$ is $k$.
\end{proof}

\begin{figure}
    \centering
    \begin{subfigure}[b]{0.45\columnwidth}
        \centering
        \includegraphics[width=0.8\linewidth]{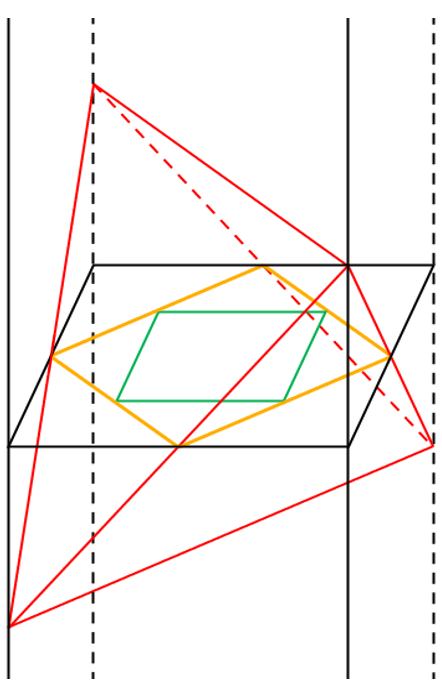}
        \caption{}
    \end{subfigure}
    \begin{subfigure}[b]{0.45\columnwidth}
        \centering
        \includegraphics[width=0.8\linewidth]{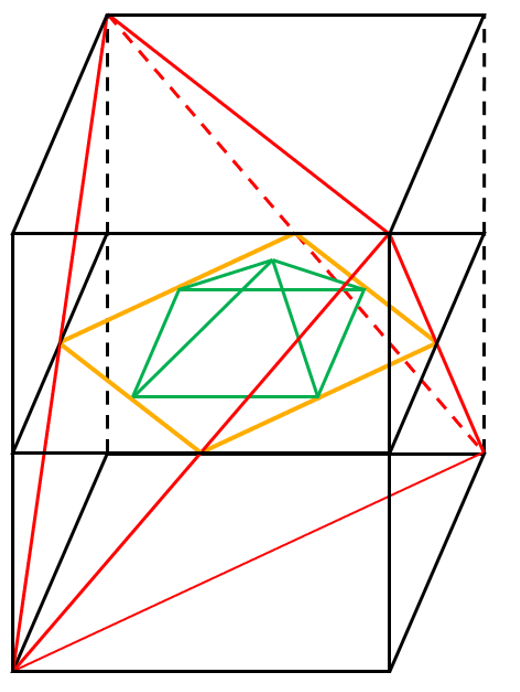}
        \caption{}
    \end{subfigure}
    \caption{Illustration of the intermediate simplex problem in Theorem~\ref{thm:reduction}. (a) The inner polytope (green) is not full-dimensional, and the outer polytope (black) is unbounded. (b) The inner polytope (green) is full-dimensional, and the outer polytope (black) is bounded.
    The intermediate simplex and its intersection with the affine span of the inner polytope are represented in red and orange, respectively.
    }
    \label{fig1}
\end{figure}

The proof of existence of the matrix $\bar{\mathbf{C}}$ in Theorem~\ref{thm:reduction} could have been done using the topological properties of $\mathbb{R}^k$.
More precisely, we could extend $\bar{\mathcal{B}}$ by simply choosing $k-r$ random vertices from a sufficiently small ball centred at its barycentre.
However, this could have implied unbounded bit lenghts for describing the added vertices.
Our selection of bounding hyperplanes for $\bar{\mathcal{A}}_b$ and the specific construction of $\bar{\mathcal{B}}_b$ ensure that the extra vertices can be specified using a polynomial number of bits, as shown in the following lemma.

\begin{lemma} \label{lemma:bit length}
    Let $\mathbf{C}$ be a rank-$r$ $m\times n$ nonnegative matrix. The bit length of entries in $\bar{\mathbf{C}}$ as defined in Theorem~\ref{thm:reduction} is poly$(b,m,n)$, where $b$ is the bit length of entries in $\mathbf{C}$.
\end{lemma}
\begin{proof}

    Since for each $i$, $\bar{h}_i$ only has to be a positive lower bound of $h_i$, it can be rounded down to a positive number when constructing $\bar{\mathcal{B}}_b$. Therefore, only the orders of magnitude of $\bar{h}_i$s determine the number of extra bits needed for defining $\bar{\mathbf{C}}$.
    
    Assuming that the entries of $\mathbf{C}$ are rational and defined with $b$ bits, the entries of the rank factors $A$ and $B$ are, without loss of generality, rational and defined with poly$(b,m,n)$ bits. Therefore, the vertices of $\bar{\mathcal{B}}$, the point $\bm{c}^1$ and the facets of $\bar{\mathcal{A}}$, as in Theorem~\ref{thm:reduction}, are defined with poly$(b,m,n)$ bits.
    Furthermore, it follows that $\mathcal{S}_{11}$ and $\mathcal{S}_{21}$, defined in the proof of Theorem~\ref{thm:reduction}, can be efficiently chosen such that the same bound on bit length applies to the vertices of the former and facets of the latter.
    For example, $\mathcal{S}_{11}$ can be found by choosing $r$ approximately evenly spaced unit vectors $\{\bm{f}^i\}$ that have poly$(b,m,n)$ bit representations. Then, for each $\bm{f^i}$, a vertex of $\mathcal{S}_{11}$ is given by $\bm{c}^1+f_i\bm{f}^i$, where $f_i$ is chosen such that $\bm{c}^1+f_i\bm{f}^i\in \bar{\mathcal{B}}$. The inclusion of $\bm{c}^1+f_i\bm{f}^i$ in $\bar{\mathcal{B}}$ can be decided by solving a linear program which has a number of constraints and parameters linear in the number of vertices of $\bar{\mathcal{B}}$. Therefore, since the input data to this linear program has poly$(b,m,n)$ bit length, a solution for $f_i$ with poly$(b,m,n)$ bit length exists and can be found using a binary search.

    We will show that each stage of the computation of $\bar{h}_i$s requires only poly$(b,m,n)$ bits. 
    The distance $d_1$ from $\bm{c}^1$ to $\partial\mathcal{S}_{11}$ is given by the distance to the closest facet of $\partial\mathcal{S}_{11}$.
    Given that $\mathcal{S}_{11}$ is $(r-1)$-dimensional, only $r-1$ vertices define each of its bounding hyperplanes. 
    Let $\{\bm{w}^1,...,\bm{w}^{r-1}\}$ be the subset of vertices of $\mathcal{S}_{11}$ in one of the bounding hyperplanes. 
    Then, to find the normal vector, define the matrix $W$ as $W_{i,:}=\bm{w}^1-\bm{w}^{i+1}$ for $i=1,...,r-2$. 
    Gaussian elimination can be used to find an unnormalized vector $\bm{n}$ in the subspace of $\operatorname{span}(\bar{\mathcal{B}})$. 
    This yields a rational vector and only requires poly$(b,m,n)$ bits.

    The distance from $\bm{c}^1$ to the hyperplane is given by $(\bm{c}^1-\bm{w}^1)^{\top}\bm{n}/|\bm{n}|$. The inner product of two rational vectors each defined with poly$(b,m,n)$ bits will have a bit length of poly$(b,m,n)$. Therefore, the minimum and maximum possible orders of magnitude of such an inner product are $2^{-\operatorname{poly}(b,n,m)}$ and $2^{\operatorname{poly}(b,n,m)}$, respectively.

    Now we will consider the maximum distance $d_2$ between two vertices of $\mathcal{S}_{21}$. 
    The intersection of only $r-1$ hyperplanes of $\bar{\mathcal{B}}_{\bar{\mathcal{A}}}$ define a vertex.
    This point can be found by solving the system of linear equations $H\bm{x} = \bm{o}$, where $H$ is a rational matrix of bit length poly$(b,m,n)$, and $\bm{o}$ is a rational vector whose entries are the offsets of the hyperplanes, also of bit length poly$(b,m,n)$. 
    As shown in Ref.~\cite{Storjohann2015}, inverting a rational matrix increases the bit length only polynomially. 
    It follows that the solution $|\bm{x}^*|$ has a maximum order of magnitude of $2^{\operatorname{poly}(b,m,n)}$.
    Furthermore, we can apply the same reasoning to the polytopes $\mathcal{S}_{2i}$ as defined in the proof of Theorem~\ref{thm:reduction}.

    Putting this together and using Eq.~\eqref{eq: htilde}, the minimum order of magnitude of $\bar{h}_1$ is $2^{-\operatorname{poly}(b,n,m)}$.
    However, we need to also take into consideration the iterations in the proof of Theorem~\ref{thm:reduction}. 
    Starting from $\bar{\mathcal{B}}$, each iteration $\bar{\mathcal{B}}_i$ is constructed by adding a vertex of the form $\bm{c}_i+\bar{h}_{i-1}\bm{h}^{i-1}$ to $\bar{\mathcal{B}}_{i-1}$ and taking the convex hull.

    Consider the second iteration, in which we find the distance of the barycentre $\bm{c}^2$ of $\bar{\mathcal{B}}_2$ to $\partial\mathcal{S}_{1,2}$, where $\bar{\mathcal{B}}_2=\operatorname{conv}(\bar{\mathcal{B}}\cup \bm{c}^1+\bar{h}_1\bm{h}^1)$, and $\mathcal{S}_{1,2}$ is a simplex contained in $\bar{\mathcal{B}}_2$ containing $\bm{c}^2$. This can simply be chosen as $\mathcal{S}_{1,2}=\operatorname{conv}(\mathcal{S}_{11}\cup \bm{c}^1+\bar{h}_1\bm{h}^1)$
    Given that $\bar{\mathcal{B}}$ has $m$ vertices, and assuming $\bar{h}_1$ is sufficiently small, this distance is given by $\bar{h}_1/m$, so its order of magnitude is at least $2^{-\operatorname{poly}(b,n,m)}$. 
    Then, by the same reasoning, the minimum distance given in each iteration is also of this form, where division by $m+i$ decreases the exponent by O(log$(m+i)$). $i$ has a maximum value of $k-r$, so the minimum order of magnitude of $\bar{h}_i$s is $2^{-\operatorname{poly}(k,l,n,m)}$.

    It follows that $\bar{h}_i$s can be rounded down such that their bit lengths are at most $\operatorname{poly}(b,k,n,m)$. Therefore, the bit lengths of the entries of $\bar{\mathbf{C}}$ are also $\operatorname{poly}(b,k,n,m)$. Since $k\leq r^2$, which is shown in Refs.~\cite{Selby_2024, Gitton_2022, schmid_2021}, and $r\leq m,n$, the bit lengths are of the form $\operatorname{poly}(b,n,m)$.
\end{proof}

\begin{corollary} \label{cor: polynomial}
    The reduction in Theorem~\ref{thm:reduction} can be computed in $\operatorname{poly}(b,n,m)$ time.
\end{corollary}

\begin{proof}

    Firstly, the simplexes $\mathcal{S}_{1i}$ and $\mathcal{S}_{2i}$ for each iteration can be constructed in polynomial time. Then, we note that the H-representaion of a simplex can be computed in polynomial time from its V-representation and vice versa. Therefore, the distances $d_{1i}$ and $d_{2i}$ can be computed in runtimes that are polynomial in the bit length required for the desired precision of these distances.

    Since this bit length is bounded by $\operatorname{poly}(b,n,m)$, the complexity of the whole procedure given in the proof of Theorem~\ref{thm:reduction} has a similar bound.
\end{proof}

In his seminal work, Moitra~\cite{moitra} showed through an explicit algorithm that
the complexity of deciding whether the NNR of an $m\times n$ nonnegative matrix is at most $k$ is poly$(b,m,n)^{O(k^2)}$.
Here, as before, $b$ determines the bit length of the entries of the input.
Therefore, using Corollary~\ref{cor: polynomial}, and given that $r\leq m,n$ and $k\leq r^2$, we have the following result.
\begin{corollary}\label{Corollary: complexity}
    The complexity of deciding if a rank-$r$ COPE $\mathbf{C}$ admits a noncontextual ontological model of inner dimension $k$ is at most \smash{poly$(b,m,n)^{O(k^2)}$}, where the entries of $\mathbf{C}$ are given with a bit length $b$.
\end{corollary}
Corollary~\ref{Corollary: complexity} provides an upper bound on the complexity of deciding if the ENNR of a matrix is at most $k$, and Corollary~\ref{cor:lower bound} provides a lower bound on the complexity of deciding if the ENNR is at least $r$. 
However, a lower bound in $k$ remains elusive.

The constructive proofs used in Theorem~\ref{thm:reduction} and Corollary~\ref{cor: polynomial} lead to an algorithm for finding the smallest noncontextual model as follows.

\begin{algorithm} [H]                    
\caption{Finding the ENNR of a COPE\\
\textbf{Input:} A rank-$r$ COPE matrix $\mathbf{C}$.\\
\textbf{Output:} The ENNR of $\mathbf{C}$.}          
\label{alg:ENNR}                          
\begin{algorithmic} [1]                   % enter the algorithmic environment
\STATE Compute a real rank factorization $\mathbf{C}=AB$.
\STATE Set $k=r$ and $ENER=0$.
\STATE Construct the matrices $\bar{A}$ and $\bar{B}$, where the rows of $\bar{A}$ and the columns of $\bar{B}$ are the trivial embeddings of the rows of $A$ and columns of $B$, respectively, into $\mathbb{R}^k$. 
\STATE Add $2(k-r)$ rows of length $k$ to the bottom of $\bar{A}$, of the form $-\bm{h}^i/2+\bm{u}/2$ for $1\leq i \leq k-r$ and $\bm{h}^i/2+\bm{u}/2$ for $k-r+1\leq i \leq 2(k-r)$ to construct $\bar{A}_b$, where $\bm{h}^i_j=\delta_{ij}$ for $1\leq j \leq k$.
\STATE Set $i=1$ and let $\mathcal{\bar{B}}_1=\mathcal{\bar{B}}$.
\STATE Compute $\bm{c}^i$, the mean of the vertices of $\mathcal{\bar{B}}_i$, and find a simplex $\mathcal{S}_{1i}$ that contains $\bm{c}^i$ and is contained in $\mathcal{\bar{B}}_i$. To find such a simplex efficiently, see Lemma~\ref{lemma:bit length}.
\STATE Find a simplex $\mathcal{S}_{2i}$ that contains $\bar{\mathcal{B}}_{\bar{\mathcal{A}}i}:=\operatorname{span}(\bar{\mathcal{B}}_i)\cap\bar{\mathcal{A}}$. This can be done by finding a simplex which is contained in the polar dual of $\bar{\mathcal{B}}_{\bar{\mathcal{A}}i}:=\operatorname{span}(\bar{\mathcal{B}}_i)\cap\bar{\mathcal{A}}$ and contains the origin, then taking the polar dual of that simplex.
\STATE Compute $d_{1i}$ as the distance from $\bm{c}^i$ to $\partial\mathcal{S}_{1i}$, and $d_{2i}$ as the diameter of $\mathcal{S}_{2i}$, both to only the precision of the leading respective bit. Round down $d_{1i}$ to the previous power of 2, and round up $d_{2i}$ to the next power of 2. Then, compute $\bar{h}_i=d_{1i}/d_{2i}$.
\STATE Construct $\bar{\mathcal{B}}_{i+1}=\operatorname{conv}(\bar{\mathcal{B}}_i\cup \bm{v}^i)$, where $\bm{v}^i = \bm{c}^i + \bar{h}_i\bm{h}^i$ and move onto the next iteration.
\STATE If $i < k-r$, set $i=i+1$ and go to step 6.
\STATE Set $\bar{\mathcal{B}}_b=\bar{\mathcal{B}}_{k-r}$ and compute $\mathbf{\bar{C}}:=\bar{A}_b\bar{B}_b$.
\STATE Decide if the NNR of $\mathbf{\bar{C}}$ equals $k$, which can be done using Moitra's algroithm. If yes, the ENNR equals $k$, set $ENNR=k$, and abort. Else if $k\leq r^2$, set $k=k+1$ and go to step 3.
\end{algorithmic}
\end{algorithm}

Indeed, using the algorithm above, if the final value of $ENNR$ is zero it means that an ENNR is not found, hence an ENMF does not exist.
Our algorithm uses a subroutine that computes NNR.
Moitra's algorithm is so far the best for this purpose, and an expert reader may notice that its first step is to guess the ranks of the nonnegative matrix factors. Therefore, one may ask if these ranks could be chosen to equal the rank of the input matrix to attempt to find an ENNR of a specified inner dimension $k$.
The limitation of Moitra's algorithm in this case is that it searches for NMFs with an additional restriction, as outlined in Appendix~\ref{section:Moitra}. This restriction does not rule out the simplest NMF, but, as shown in Appendix~\ref{section:Moitra}, could rule out the simplest ENMF.

\section{Comparison to Simplex Embeddability}\label{sec:comparison}

As we mentioned in Sec.~\ref{subsec:OM}, every NMF  $\mathbf{C}= \mathbf{R}\mathbf{E}$ of a COPE matrix, after an appropriate rescaling of its factors, corresponds to an ontological model. 
We have also seen that any ENMF of $\mathbf{C}$, which is simply an NMF in which all factors have equal ranks, can be written as $\mathbf{C}= \bar{A}\bar{G}\bar{G}^{-1}\bar{B}$ so that the factor matrices correspond to the nested polytope $\bar{A}\mapsto \bar{\mathcal{A}}$, $\bar{G}\mapsto \bar{\mathcal{G}}$, and $\bar{B}\mapsto \bar{\mathcal{B}}$ such that $\bar{\mathcal{G}}$ is a $(k-1)$-simplex and $\bar{\mathcal{B}} \subseteq \bar{\mathcal{G}} \subseteq \bar{\mathcal{A}}$, as shown in Lemma~\ref{lemma:cone}.

Alternatively, it was shown in Refs.~\cite{Shahandeh2021,simplexembeddability} that noncontextuality has a geometrical interpretation known as \textit{simplex embeddability}. 
Namely, a GPT admits a noncontextual ontological model if and only if its state and effect spaces can be embedded into a simplex and a hypercube, respectively.
As demonstrated in Section.~\ref{section 2}, a GPT model can be easily found from the COPE of a prepare-measure experiment. 
In this section, we characterize the relationship between our geometric interpretation of noncontextuality and simplex embeddability.

We denote the polytopes of states and effects by $\mathcal{Q}$ and $\mathcal{E}$, respectively, which live in an inner product space $(\mathfrak{V},\langle\cdot,\cdot\rangle_{\mathfrak{V}})$. Then, simplex embeddability is defined as follows.

\begin{definition}[Simplex Embeddability~\cite{simplexembeddability}.] \label{embeddabilitydefn}
    A GPT with state and effect spaces $\mathcal{S}$ and $\mathcal{E}$ over a vector space $\mathfrak{V}$ with inner product $\langle \cdot ,\cdot\rangle_{\mathfrak{v}}$
      is simplex-embeddable if and only if there exists a $(k-1)$-simplex $\mathcal{W}_k$ (whose affine span does not contain the origin) in a $k$-dimensional inner product space $\mathfrak{W}_k$
       and its dual hypercube
      $\mathcal{W}_{k}^*$ (including the zero and unit vectors), and
      a pair of linear maps $\iota,\kappa : \mathfrak{V}\to \mathfrak{W}_k$ such that $\iota(\mathcal{S}) \subseteq  \mathcal{W}_{k}$, $\kappa(\mathcal{E}) \subseteq \mathcal{W}_{k}^*$, and $\langle\kappa({\boldsymbol{e}}),\iota({\boldsymbol{s}})\rangle_\mathfrak{W} = \langle \boldsymbol{e},\boldsymbol{s}\rangle_\mathfrak{v}$ for all GPT states and effects.
    \end{definition}

This formulation, like ours, requires finding a single polytope, specifically a simplex.
Furthermore, Ref.~\cite{simplexembeddability} shows that after the simplex embedding is found, there exists a linear transformation $T$ such that its adjoint, $T^\dagger$ maps the vertices $\mathcal{W}_k$ to the positive unit vectors in $\mathbb{R}^k$, and $\langle T(\bm{e}),T^\dagger(\bm{s})\rangle=\langle\bm{e},\bm{s}\rangle$ for all pairs of states and effects.
We can identify our matrix $\bar{G}$ with the transformation $T\circ\kappa$, and $\bar{G}^{-1}$ with $T^\dagger\circ\iota$.
To see the connection between $\mathcal{\bar{G}}$ and $\mathcal{W}_k$, let $\iota^1$ be the trivial embedding into $\mathbb{R}^k$, and define $\kappa^2$ and $\iota^2$ such that $\kappa=\kappa^2\circ\iota^1$ and $\iota=\iota^2\circ\iota^1$. Then, $\iota^2(\mathcal{\bar{G}})=\mathcal{W}_k$.

Ref.~\cite{linearprogram} suggests using a linear program to decide if the simplex embeddability conditions can be satisfied. Here, each solution to the linear program is a matrix that represents the product of a pair of linear maps $\iota,\kappa$. Then, the dimensionality of the codomains of the associated maps $\iota,\kappa$ to a solution instance is given by its nonnegative rank.
The smallest noncontextual model is found by searching over the set of solutions to find the minimum nonnegative rank in that set. 
Since the nonnegative rank is a non-smooth function of matrices, the number of times a nonnegative rank of a matrix needs to be computed could be very large, perhaps even exponential in the number of input states and effects.
The key utility of our construction is that it allows us to compute the ontic size $d$ by only computing the nonnegative rank of $d+1$ matrices.
Using our construction, the size $d$ of the smallest noncontextual model is found by testing increasing ontic sizes starting at the minimum GPT dimension $r$. For each ontic size, the nonnegative rank of only one matrix needs to be computed, so the total number of nonnegative rank computations needed is $d-r$.

\section{Ontic Size Separation} \label{Ontic gap}

We have seen that the problem of finding the smallest noncontextual ontological model for a given COPE matrix is equivalent to the problem of finding the smallest (not necessarily noncontextual) ontological model of a different nonnegative matrix.
We may therefore ask whether a noncontextual ontological model, when exists, is the smallest ontological model.
In other words,
is there a gap between the ENNR and the NNR whenever an ENMF exists?

In this section, we present a COPE matrix that shows such a gap, indicating that the ENNR cannot be trivially obtained from the NNR and vice versa.

We show this for a $5\times5$ rank-3 COPE matrix whose first column is given by,
\begin{equation}
    \mathbf{C}^1_{:1}=\frac{1}{10}\left(
\begin{array}{ccccc}
 2 \left(\sqrt{5}-1\right) \\
 8-2 \sqrt{5}\\
 2 \left(\sqrt{5}-1\right)\\
 3-\sqrt{5}\\
 3-\sqrt{5}\\
\end{array}
\right),
\end{equation}
and each subsequent column is given by permuting the entries the previous column up by one. This COPE also admits the following real factorization,
\begin{equation} \label{C1}   \mathbf{C}^1:=A_{\mathcal{U}}B_{\mathcal{I}}, 
\end{equation}
which has associated inner and outer polygons (pentagons), shown in Fig.~\ref{fig:square}, $\mathcal{B_I}$ and $\mathcal{A_U}$, respectively, as in Lemma~\ref{lemma:lemma4}. These are defined by constructing a pentagon $\mathcal{I}$ and the smallest regular 4-simplex $\mathcal{U}$ that contains it.

We begin with the bounding inequalities of a regular pentagon, given by $K \bm{x} + \bm{p} \geq 0$ wherein,
\begin{equation}
\begin{split}
    K= \begin{pmatrix}
    \sin\left(\frac{2\pi}{5}\right) & \cos\left(\frac{2\pi}{5}\right) \\
    \sin\left(\frac{4\pi}{5}\right) & \cos\left(\frac{4\pi}{5}\right) \\
    \sin\left(\frac{6\pi}{5}\right) & \cos\left(\frac{6\pi}{5}\right) \\
    \sin\left(\frac{8\pi}{5}\right) & \cos\left(\frac{8\pi}{5}\right) \\
    \sin\left(\frac{10\pi}{5}\right) & \cos\left(\frac{10\pi}{5}\right)
    \end{pmatrix},
\end{split}
\end{equation}
and
\begin{equation}
    \bm{p}=
    \frac{\sqrt{5}-1}{10}\begin{pmatrix}
        1 & 1 & 1 & 1 & 1
    \end{pmatrix}^\top.
\end{equation}
The vertices of this are given by,
\begin{equation} \label{xvertices}
\begin{split}
    \bm{x}^{i} &=  \frac{\sqrt{5}-1}{10}\begin{pmatrix}
    \cos\left(\frac{2\pi i}{5}\right) \tan\left(\frac{2\pi}{10}\right) + \sin\left(\frac{2\pi i}{5}\right) \\
    \sin\left(\frac{2\pi i}{5}\right) \tan\left(\frac{2\pi}{10}\right) - \cos\left(\frac{2\pi i}{5}\right)
    \end{pmatrix},
\end{split}
\end{equation}
for $i \in \{1,2,...,5\}$.

Now, we can use this to construct $\mathcal{I}$ and $\mathcal{U}$. The points $\{\bm{y}^{i}\}$ given by $\bm{y}^{i}=(2/5)(K\bm{x}^{i}+\bm{p})$ form a regular pentagon embedded in $\mathbb{R}^{5}_{\geq0}$. 
Let $\mathcal{I}$ and $\mathcal{U}$ be the convex hull of $\{\bm{y}^{i}\}$ and the regular 4-simplex whose vertices are given by the columns of $([\sqrt{5}-1]/5)\mathds{1}_{5\times5}$, respectively.
Note that $\mathcal{U}$ can also be defined as,
\begin{align}
    \mathcal{U}:=\left\{\bm{u}\in\mathbb{R}^5|\bm{u}\geq0, \bm{u}^\top\left(\frac{10}{\sqrt{5}-1}\bm{p}\right)=\frac{\sqrt{5}-1}{5}\right\}.
\end{align}
Then, for any real vector $\begin{pmatrix}
    a & b
\end{pmatrix}^\top\in\mathbb{R}^2$,
\begin{equation}
\begin{split}
    & \begin{pmatrix}
    a & b
\end{pmatrix} K^\top
\left(\frac{10}{\sqrt{5}-1}\bm{p}\right) = \\
    &\sum_{i=1}^5 \left( a \cos\left(\frac{2\pi i}{5}\right) + b \sin\left(\frac{2\pi i}{5}\right) \right)= 0.
\end{split}
\end{equation}
Therefore, given,
\begin{equation}
    \frac{2}{5}\bm{p}^\top\left(\frac{10}{\sqrt{5}-1}\bm{p}\right)=\frac{\sqrt{5}-1}{5},
\end{equation}
it is clear that $\mathcal{I}\subset \mathcal{U}$.

Let $\mathcal{P}_{\mathcal{U}}$ be the projection of $\mathcal{U}$ onto the affine span of $\mathcal{I}$, so $\mathcal{P}_{\mathcal{U}}$ is another regular pentagon.

It can be shown that $\mathcal{P}_{\mathcal{U}}$ is a $\times(\sqrt{5}-1)/2$ scaling of $\mathcal{I}$ about its centre, and this scaling factor $(\sqrt{5}-1)/2\approx0.618$ is derived in Appendix~\ref{size comparison}.
Then, we choose $\mathcal{A}_{\mathcal{U}}$ and $\mathcal{B}_{\mathcal{I}}$ to be nested pentagons where $\mathcal{A}_{\mathcal{U}}$ is a scaling of $\mathcal{B}_{\mathcal{I}}$ by a scale factor $(\sqrt{5}-1)/2$ about its centre. Then, the corresponding matrices $A_{\mathcal{U}}$ and $B_{\mathcal{I}}$ are given by,
\begin{align}
    (A_{\mathcal{U}})_{i,:}=\begin{pmatrix}
        \frac{2}{5-\sqrt{5}} & K_{i,:}
    \end{pmatrix},\\
    (B_{\mathcal{I}})_{:i}=\begin{pmatrix}
        \frac{5-\sqrt{5}}{10} & \bm{x}^{i\top}
\end{pmatrix}^\top.
\end{align}
We have $(A_{\mathcal{U}}B_{\mathcal{I}})_{:i}=K\bm{x}^i+2/(\sqrt{5}-1)\bm{p}$, which has the form of the slack of a vertex of a regular pentagon with respect to the facets of a factor of $(\sqrt{5}-1)/2$ scaling of that pentagon.

Next, to see the existence of a noncontextual model, we construct the infinite cone $\bar{\mathcal{A}}_{\mathcal{U}}\subset\mathbb{R}^5$ from $\mathcal{A}_{\mathcal{U}}$, and trivially embed the polytope $\bar{\mathcal{B}}_{\mathcal{I}}$ into $\mathbb{R}^5$ as in Lemma \ref{lemma:cone}. It follows from the fact that $\mathcal{I}\subset\mathcal{U}$ that there exists a nested simplex $\mathcal{G}$ between $\bar{\mathcal{B}}_{\mathcal{I}}$ and $\bar{\mathcal{A}}_{\mathcal{U}}$. Therefore, Lemma~\ref{lemma:cone} can be used to show that there is a noncontextual model of ontic size five for $\mathbf{C}^1$, given by $\mathbf{C}^1=\mathbf{A}^1\mathbf{B}^1$, where,
\begin{equation}
    \begin{aligned}
        \mathbf{A}^1&=
\left(
\begin{array}{ccccc}
 \frac{5-\sqrt{5}}{10} & \frac{\sqrt{5}}{5} & \frac{5-\sqrt{5}}{10} & 0 & 0 \\
 0 & \frac{5-\sqrt{5}}{10} & \frac{\sqrt{5}}{5} & \frac{5-\sqrt{5}}{10} & 0 \\
 0 & 0 & \frac{5-\sqrt{5}}{10} & \frac{\sqrt{5}}{5} & \frac{5-\sqrt{5}}{10} \\
 \frac{5-\sqrt{5}}{10} & 0 & 0 & \frac{5-\sqrt{5}}{10} & \frac{\sqrt{5}}{5} \\
 \frac{\sqrt{5}}{5} & \frac{5-\sqrt{5}}{10} & 0 & 0 & \frac{5-\sqrt{5}}{10} \\
\end{array}
\right),\\
        \mathbf{B}^1&=\left(
\begin{array}{ccccc}
 0 & \frac{5-\sqrt{5}}{10} & \frac{\sqrt{5}}{5} & \frac{5-\sqrt{5}}{10} & 0 \\
 \frac{5-\sqrt{5}}{10} & \frac{\sqrt{5}}{5} & \frac{5-\sqrt{5}}{10} & 0 & 0 \\
 \frac{\sqrt{5}}{5} & \frac{5-\sqrt{5}}{10} & 0 & 0 & \frac{5-\sqrt{5}}{10} \\
 \frac{5-\sqrt{5}}{10} & 0 & 0 & \frac{5-\sqrt{5}}{10} & \frac{\sqrt{5}}{5} \\
 0 & 0 & \frac{5-\sqrt{5}}{10} & \frac{\sqrt{5}}{5} & \frac{5-\sqrt{5}}{10} \\
\end{array}
\right).
    \end{aligned}
\end{equation}

Before proving the ontic size separation, we present a geometric lemma, which we will use for the main proof.
\begin{lemma} \label{geometric lemma}
Let $\{\bm{v}^{i}\}_{i=0}^4$ be the set of vertices of a convex pentagon. 
For each $\bm{v}^{i}$, let the triangle $\mathcal{T}^{i}$ be $\operatorname{conv}\{\frac{\bm{v}^{i-1}+\bm{v}^{i}}{2},\bm{v}^{i},\frac{\bm{v}^{i+1}+\bm{v}^{i}}{2}\}$, where superscript additions are $\mod 5$.
\begin{figure}[h]
    \centering
    % Subfigure B
    \begin{subfigure}{0.22\textwidth}
        \centering
        \begin{tikzpicture}[scale=1]

    % Define Outer Pentagon Vertices
    \coordinate (O1) at (0.726543, -1.000000);
    \coordinate (O2) at (1.175571, 0.381966);
    \coordinate (O3) at (0.000000, 1.236068);
    \coordinate (O4) at (-1.175571, 0.381966);
    \coordinate (O5) at (-0.726543, -1.000000);

    % Define Inner Pentagon Vertices
    \coordinate (I1) at (0.449028, -0.618034);
    \coordinate (I2) at (0.726543, 0.236068);
    \coordinate (I3) at (0.000000, 0.763932);
    \coordinate (I4) at (-0.726543, 0.236068);
    \coordinate (I5) at (-0.449028, -0.618034);

    \coordinate (G1) at (0.951057, -0.309017);
    \coordinate (G2) at (0.5877855, 0.809017);
    
    \draw[thick, red] (O1) -- (G1) -- (G2) -- (O3) -- (O4) -- (O5) -- cycle;
    \draw[thick, black] (G1) -- (O2) -- (G2);

    \draw[thick, Green] (I1) -- (I2) -- (I3) -- (I4) -- (I5) -- cycle;
    \draw[thick, dashed, black] (G2) -- (O3) -- (O4) -- (O5) -- (O1) -- (G1);

    \end{tikzpicture}
    \caption{}
        \label{fig:square}
    \end{subfigure}
    \hfill
    % Subfigure A
    \begin{subfigure}{0.22\textwidth} 
        \centering
        \begin{tikzpicture}[scale=1]

    % Define Outer Pentagon Vertices
    \coordinate (O1) at (0.726543, -1.000000);
    \coordinate (O2) at (1.175571, 0.381966);
    \coordinate (O3) at (0.000000, 1.236068);
    \coordinate (O4) at (-1.175571, 0.381966);
    \coordinate (O5) at (-0.726543, -1.000000);

    % Define Inner Pentagon Vertices
    \coordinate (I1) at (0.449028, -0.618034);
    \coordinate (I2) at (0.726543, 0.236068);
    \coordinate (I3) at (0.000000, 0.763932);
    \coordinate (I4) at (-0.726543, 0.236068);
    \coordinate (I5) at (-0.449028, -0.618034);

    \coordinate (G1) at (0.85065 , -0.618035);
    \coordinate (G2) at (-0.85065 , -0.618035);
    \coordinate (G3) at (-0.64984,  0.76393);
    \coordinate (G4) at (0.64984,  0.76393);
    % Draw the outer pentagon blue
    \draw[thick, black] (O1) -- (O2) -- (O3) -- (O4) -- (O5) -- cycle;
    \draw[thick, red] (G1) -- (G2) -- (G3) -- (G4) -- cycle; % Nested quadrilateral

    % Draw the inner pentagon (red)
    \draw[thick, Green] (I1) -- (I2) -- (I3) -- (I4) -- (I5) -- cycle;

    \end{tikzpicture}
        \caption{}
        \label{fig:quadrilateral}
        \label{fig:pentagon}
    \end{subfigure}
    
    \caption{Inner (green) and outer (black) pentagons with nested polytopes (red) (a) $\mathcal{G}^2$ and (b) $\mathcal{G}^3$.}
    \label{fig:twoshapes}
\end{figure}
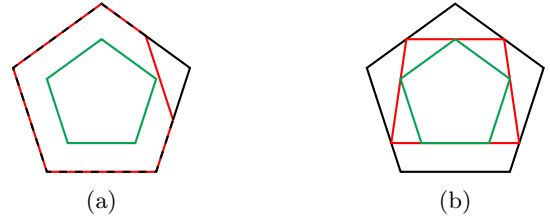
Then it is not possible to inscribe a quadrilateral in the pentagon such that it intersects the interiors of each $\mathcal{T}^{i}$.
\end{lemma}
\begin{proof}
Since all vertices of the pentagon are convexly independent, the interior of each triangle is convexly independent from the other triangles. Therefore, a line can only intersect the interiors of at most two triangles. Assume there is an inscribed quadrilateral with vertices $\{\bm{u}^i\}_{i=0}^3$ that intersects the interiors of all the triangles. Then, at least one vertex of the quadrilateral, $\bm{u}^{0}$ must be in a subset of one of the triangles, say $\mathcal{T}^{1}$, given by $\mathcal{T}^{1}\setminus \operatorname{conv}\{(\bm{v}^{0}+\bm{v}^{1})/2,(\bm{v}^{1}+\bm{v}^{2})/2\}$. If the quadrilateral intersects the interior of $\mathcal{T}^{2}$, then since $\operatorname{conv}\{\bm{u}^{0}, \bm{u}^{1}\}$ can only intersect up to two triangles, $\bm{u}^{1}$ must be in $\mathcal{T}^{2}\setminus \operatorname{conv}\{(\bm{v}^{2}+\bm{v}^{1})/2,(\bm{v}^{3}+\bm{v}^{2})/2\}$. Similarly, the quadrilateral would have to have another vertex, $\bm{u}^{2}$, in $\mathcal{T}^{3}\setminus \operatorname{conv}\{(\bm{v}^{2}+\bm{v}^{3})/2,(\bm{v}^{3}+\bm{v}^{4})/2\}$ to intersect the interior of $\mathcal{T}^{3}$, and so on, i.e., one quadrilateral vertex for each $\mathcal{T}^{i}$, which is not possible.
\end{proof}
We are now ready to prove our second claim.
\begin{proposition} \label{k=4}
There is no ENMF of $\mathbf{C}^1$ of inner dimension four.
\end{proposition}
\begin{proof}
Assume that there is an ENMF of $\mathbf{C}^1$ of inner dimension four.
Then there is a nested quadrilateral $\mathcal{G}^{1}$ between $\mathcal{B}_{\mathcal{I}}$ and $\mathcal{A}_{\mathcal{U}}$ such that there is a rank-3 nonnegative matrix $\mathbf{E}^{1}$ that maps the vertices of $\mathcal{G}^{1}$ to the vertices of $\mathcal{B}_{\mathcal{I}}$.
Note that for any larger polytope $\mathcal{G}^{2}$ that contains $\mathcal{G}^{1}$, there is a nonnegative rank-3 matrix $\mathbf{E}^{2}=\mathbf{WE}^{1}$ that maps the vertices of $\mathcal{G}^{2}$ to $\mathcal{B}_{\mathcal{I}}$, where $\mathbf{W}$ is a nonnegative matrix that maps the vertices of $\mathcal{G}^{2}$ to $\mathcal{G}^{1}$.

Labelling the vertices of $\mathcal{A}_{\mathcal{U}}$ as $\bm{v}^{0},...,\bm{v}^{4}$ as in Lemma \ref{geometric lemma}, let $\mathcal{G}^{2}$ be $\operatorname{conv}(\{\bm{v}^{1},...,\bm{v}^{4}\}\cup\{(\bm{v}^{0}+\bm{v}^{1})/2,(\bm{v}^{4}+\bm{v}^{0})/2\})$, as shown in Fig.~\ref{fig:square}.
Note that this polytope is also given by removing one of the triangles $\mathcal{T}^i$ described in Lemma~\ref{geometric lemma}.
Then, by Lemma~\ref{geometric lemma} and the rotational symmetry of the pentagon, any nested quadrilateral $\mathcal{G}^{1}$ is contained in $\mathcal{G}^{2}$, up to a $2n\pi/5$ rotation, where $n$ is an integer.

Now, we use the algorithm in the proof of Lemma \ref{lemma:algorithm} to find such a matrix $\mathbf{E}^{2}$ by solving $G^{2}\mathbf{E}^{2}=\mathcal{A}_{\mathcal{I}}$.  Using the linear program described in the proof of Lemma \ref{lemma:algorithm}, we find that the minimum value of $f(D)$, as defined in Lemma \ref{lemma:algorithm}, is about $0.146>0$, indicating that there is no solution for $\mathbf{Q}^{2}$. We also find that the maximum of the dual linear program is the same. Since any value of the dual linear program lower bounds the minimum of the primal linear program, this proves that the minimum of the primal problem is positive. Thus, any matrix $E$ that satisfies $G^{2}E=\mathcal{A}_{\mathcal{I}}$ must have negative elements. We outline the calculation of this objective value in Appendix C.
\end{proof}
The final step is now to show that the COPE $\mathbf{C}^1$ admits an ontological model of dimension four. This model is necessarily contextual.
\begin{lemma}
There exists an ontological model of $\mathbf{C}^1$ of ontic size four.
\end{lemma}
\begin{proof}
The existence of a nested quadrilateral $\mathcal{G}^{3}$ as shown in Fig.~\ref{fig:quadrilateral} implies an NMF of $\mathbf{C}^1$ of inner dimension four. Constructing $G^{3}$ from $\mathcal{G}^{3}$, $\mathcal{G}^{3}\subset \mathcal{A}_{\mathcal{U}}$ implies $A_{\mathcal{U}}G^{3}$ is nonnegative. $\mathcal{B}_{\mathcal{I}}\subset \mathcal{G}^{3}$ implies there is a nonnegative matrix that maps the columns of $G^{3}$ to the columns of $B_{\mathcal{I}}$. Such an ontological model $\mathbf{C}^1=\mathbf{A}^2\mathbf{B}^2$ is given explicitly as,
\begin{equation}
\begin{split}
&\mathbf{A}^2=\left(
\begin{array}{cccc}
 0 & 1-\frac{1}{\sqrt{5}} & 0 & \frac{\left(5-\sqrt{5}\right)}{10} \\
 \frac{1}{\sqrt{5}} & \frac{1}{\sqrt{5}} & 0 & 0 \\
 1-\frac{1}{\sqrt{5}} & 0 & \frac{\left(5-\sqrt{5}\right)}{10} & 0 \\
 0 & 0 & 1-\frac{1}{\sqrt{5}} & \frac{\left(3 \sqrt{5}-5\right)}{10} \\
 0 & 0 & \frac{\left(3 \sqrt{5}-5\right)}{10} & 1-\frac{1}{\sqrt{5}} \\
\end{array}
\right),\\
&\mathbf{B}^{2}=\left(
\begin{array}{ccccc}
 \frac{2}{\sqrt{5}}-\frac{1}{2} & \frac{5-\sqrt{5}}{20} & 0 & \frac{3 \sqrt{5}-5}{10} & \frac{3 \left(5-\sqrt{5}\right)}{20} \\
 \frac{2}{\sqrt{5}}-\frac{1}{2} & \frac{3 \left(5-\sqrt{5}\right)}{20} & \frac{3 \sqrt{5}-5}{10} & 0 & \frac{5-\sqrt{5}}{20} \\
 1-\frac{2}{\sqrt{5}} & 0 & \frac{5-\sqrt{5}}{10} & 1-\frac{1}{\sqrt{5}} & \frac{1}{\sqrt{5}} \\
 1-\frac{2}{\sqrt{5}} & \frac{1}{\sqrt{5}} & 1-\frac{1}{\sqrt{5}} & \frac{5-\sqrt{5}}{10} & 0 \\
\end{array}
\right).\\
\end{split}
\end{equation}

\end{proof}
We have now seen a COPE whose minimum ontic size is four and minimum noncontextual ontic size is five. This proves that the number of ontic states of the smallest ontological model and the smallest contextual model are different in general. This implies that, when studying how the latter depends on properties of the COPE, the former cannot be directly used.

\section{Discussion and Conclusions} \label{section 5}

We introduced the equirank nonnegative matrix factorization (ENMF) problem, which is to find a nonnegative matrix factorization (NMF) wherein all components have equal ranks.
We demonstrated that this is a linear-algebraic formulation of the problem of finding a noncontextual ontological model for explaining a set of observed statistics as characterized by the rank-separation criterion.
We proved that the complexity of deciding whether COPE matrices of fixed rank admit a noncontextual model of any size is polynomial in the size of the matrices.
We achieved this through a novel geometric reformulation of the problem:
We showed that deciding and  computing an ENMF is equivalent to deciding and finding an intermediate simplex sandwiched between the preparations and outcome events polytope and dual polytope, respectively.

Our construction also resulted in a reduction from ENMF to the nonnegative rank (NNR) problem, wherein one asks whether the NNR of a nonnegative matrix equals its rank.
This led to an upper bound on the complexity of deciding and computing noncontextual ontological models in given dimensions.
More specifically, we showed that the latter is at most exponential in the required noncontextual ontic dimension.

From a practical point of view, we would like to know whether a given experiment can be modelled classically, i.e., by a noncontextual ontological model.
Furthermore, we would like to determine the minimum number of classical states required for such a model, whenever it exists, in order to estimate the classical resources necessary for emulating the experiment classically.
The geometric interpretations, as well as the constructive nature of our proofs, put forward a method for answering these questions.
For example, obtaining tight upper bounds and lower bounds on ENNR and NNR, respectively, provides an effective method for detecting contextuality in the data.
Specifically, if the lower bound on the NNR of a COPE exceeds the upper bound on its ENNR, it implies the presence of contextuality.
However, our results also provide evidence that addressing these problems is computationally demanding.
We conjecture that, contrary to common belief, detecting contextuality is computationally hard.
This is further supported by the fact that all currently known algorithms for this task have runtimes that scale exponentially with the rank of the COPE matrix, or its equivalent, the GPT dimension.
However, a formal proof of this conjecture remains elusive.

Given that finding ontological models and noncontextual ontological models have geometrical interpretations which reduce to the intermediate simplex problem, it is natural to ask whether they are genuinely different problems.
We answered this question in the affirmative by constructing an explicit example of a COPE matrix which exhibits a size gap between its smallest ontological model and smallest noncontextual model.
The roles of these minimum model sizes in explaining the quantum advantage of information processing protocols remain unclear.
Of course, the latter is only relevant if a noncontextual model exists, which seems to greatly restrict the degree of quantum advantage possible.
This can be seen by considering the maximum possible gap between the smallest GPT dimensionality and the smallest noncontextual ontic dimensionality needed to explain a COPE.
Since the former is given by the rank $r$ and the latter is at most $O(r^2)$, the advantage of noncontextually explainable protocols would be polynomially bounded.

A better understanding of the features of the COPE matrix that dictate the behavior of NNR and ENNR could make it easier to identify tasks and protocols that exhibit a quantum advantage. 
For instance, the role of NNR in classical and quantum communication complexity has been studied in Refs.~\cite{lalonde2023,heinosaari2024simpleinformationprocessingtasks}.
However, to our knowledge, no prior work has explored its connection to time or space complexity of quantum circuits. 
Furthermore, analyzing communication and computation complexities through the lens of ENNR, hence contextuality, is an interesting direction for future research.

\acknowledgements
TY acknowledges the support through the Quantum Computing Studentship funded by Royal Holloway, University of London.
FS gratefully acknowledge the financial support from the Engineering and Physical Sciences Research Council (EPSRC) through the Hub in Quantum Computing and Simulation grant (EP/T001062/1).

%\vspace{10cm}
\bibliography{reffinal.bib}
\bibliographystyle{apsrev4-1}

\begin{widetext}

\appendix

\setcounter{equation}{0}
\renewcommand{\theequation}{\thesection.\arabic{equation}}
\section{Explanation of Linear Programs}\label{app:L_optim}

Our task is to find a matrix $D$ that minimizes the sum of the moduli of negative elements of $L\bar{E}$, where $L$ is defined in Eq.~\eqref{shear}. We take $\bar{E}$ to be a $k\times n$ matrix, and $D$ to be a $(k-r)\times r$ matrix. This can be formulated as a linear program, which in general, along with its dual linear program, have the following forms.

\textbf{Standard Primal Problem:}

\[
\begin{aligned}
& \text{Minimize} && \Tilde{\bm{c}}^\top \Tilde{\bm{x}} \\
& \text{Subject to} && \Tilde{A} \Tilde{\bm{x}} \geq \Tilde{\bm{b}} \\
& && \Tilde{\bm{x}} \geq \mathbf{0}
\end{aligned}
\]

\textbf{Standard Dual Problem:}

\[
\begin{aligned}
& \text{Maximize} && \tilde{\bm{b}}^\top \Tilde{\mathbf{y}} \\
& \text{Subject to} && \Tilde{A}^\top \Tilde{\mathbf{y}} \leq \Tilde{\bm{c}} \\
& && \Tilde{\mathbf{y}} \geq \mathbf{0}
\end{aligned}
\]

Now, we will formulate our minimization task as a linear program.

\textbf{Primal Problem:}

We introduce auxiliary variables $\mathbf{S}_{ps}\geq\max{[0,(-L\bar{E})_{ps}]}$ to represent the magnitude of negative deviations. Since $D$ can have negative entries, we introduce nonnegative variables \( \mathbf{D}^+_{ij} \) and \( \mathbf{D}^-_{ij} \) such that $D_{ij}=\mathbf{D}^+_{ij}-\mathbf{D}^-_{ij}$. The primal problem is as follows,
\[
\begin{aligned}
& \text{Minimize} && \sum_{p=1}^{m} \sum_{s=1}^{k} \mathbf{S}_{ps} \\    
& \text{Subject to} \\
& &&  \mathbf{S}_{ps}+\sum_{i=1}^{k - r} \sum_{j=1}^{r} \left(  a^{j\top} \bar{E}_{:s} b^{i}_p \mathbf{D}^-_{ij} -  a^{j\top} \bar{E}_{:s} b^{i}_p \mathbf{D}^+_{ij} \right) \geq -\bar{E}_{ps} \quad \forall p = 1, \dots, m; \; l = 1, \dots, k \\
& && \mathbf{S}_{ps} \geq 0 \quad \forall p = 1, \dots, k; \; s = 1, \dots, n \\
& && \mathbf{D}^+_{ij} \geq 0 \quad \forall i = 1, \dots, m - r; \; j = 1, \dots, r \\
& && \mathbf{D}^-_{ij} \geq 0 \quad \forall i = 1, \dots, m - r; \; j = 1, \dots, r
\end{aligned}
\]
The first constraint represents $\mathbf{S}_{ps}\geq(-L\bar{E})_{ps}$, using the definition of $L$ in Eq.~\eqref{shear}.

In this case, the variable $\Tilde{\bm{x}}$ is given by $\begin{pmatrix}
    \mathbf{S}_{1,:} 
& \dots 
& \mathbf{S}_{m,:} & \mathbf{D}^+_{1,:} 
& \dots & \mathbf{D}^+_{m-r,:} 
& \mathbf{D}^-_{1,:} 
& \dots 
& \mathbf{D}^-_{m-r,:}
\end{pmatrix}^\top$, which is a concatenation of three flattened matrices.
Alternatively, we can write $\Tilde{\bm{x}}=\begin{pmatrix}
    (\Tilde{\bm{x}}^1)^\top & (\Tilde{\bm{x}}^2)^\top & (\Tilde{\bm{x}}^3)^\top
\end{pmatrix}^\top$, where 
\[
\begin{split}
\Tilde{\bm{x}}^1_q &= \mathbf{S}_{ps}, \quad 
p = (q-1\mod k) + 1, \quad s = \left\lfloor \frac{q-1}{k} \right\rfloor + 1, \quad 
q \in \{1, \dots, kn\}, \\
\Tilde{\bm{x}}^2_q &= \mathbf{D}^+_{ij}, \quad 
i = \left\lfloor \frac{q-1}{r} \right\rfloor + 1, \quad 
j = (q-1\mod r) + 1, \quad 
q \in \{1, \dots, (k - r) r\}, \\
\Tilde{\bm{x}}^3_q &= \mathbf{D}^-_{ij}, \quad 
i = \left\lfloor \frac{q-1}{r} \right\rfloor + 1, \quad 
j = (q-1\mod r) + 1, \quad 
q \in \{1, \dots, (k - r) r\}.
\end{split}
\]

$\Tilde{\bm{c}}^\top$ is given by $kn$ ones followed by $2(k-r)r$ zeroes.
$\Tilde{A}$ can be written in block matrix form as follows:
\begin{equation}
    \begin{split}
\Tilde{A} &= 
\begin{pmatrix}
\mathds{1}_{kn\times kn} & -R & R
\end{pmatrix}, \\
R_{tu} &= a^{j\top} \bar{E}_{:s} b^i_p, \quad 
t \in \{1, \dots, kn\}, \quad 
u \in \{1, \dots, (k - r) r\}, \\
p &= (t-1\mod k) + 1, \quad 
s = \left\lfloor \frac{t-1}{k} \right\rfloor + 1, \\
i &= \left\lfloor \frac{u-1}{r} \right\rfloor + 1, \quad 
j = (u-1\mod r) + 1.
\end{split}
\end{equation}

The vector \(\Tilde{\bm{b}}\) has length \(kn\),
which is given by,
\[
\Tilde{\bm{b}}=-
\begin{pmatrix}
    \bar{E}_{:1}^\top & \bar{E}_{:2}^\top 
& \dots & \bar{E}_{:n}^\top
\end{pmatrix}^\top.
\]
\textbf{Dual Problem:}

\[
\begin{aligned}
& \text{Maximize} && -\sum_{p=1}^{k} \sum_{s=1}^{n} \bar{E}_{ps} \mathbf{Y}_{ps} \\
& \text{Subject to} \\
& && \sum_{p=1}^{k} \sum_{s=1}^{n} \left( - a^{j\top} \bar{E}_{:s} b^i_p \right) \mathbf{Y}_{ps} \leq 0 \quad \forall i = 1, \dots, k - r; \; j = 1, \dots, r \\
& && \sum_{p=1}^{k} \sum_{s=1}^{n} \left(  a^{j\top} \bar{E}_{:s} b^i_p \right) \mathbf{Y}_{ps} \leq 0 \quad \forall i = 1, \dots, k - r; \; j = 1, \dots, r \\
& && \mathbf{Y}_{ps} \leq 1 \quad \forall p = 1, \dots, k; \; l = 1, \dots, n \\
& && \mathbf{Y}_{ps} \geq 0 \quad \forall p, s,
\end{aligned}
\]

which we derive using the standard relationship between primal and dual linear programs, and $\Tilde{\mathbf{y}}=
\begin{pmatrix}
\mathbf{Y}_{:1}^\top 
& \mathbf{Y}_{:2}^\top & \dots & \mathbf{Y}_{:n}^\top
\end{pmatrix}^\top$.
Defining $H^{(ij)}$ as $H^{(ij)}_{ps}:= a^{j\top} \bar{E}_{:s} b^i_p$, the first two constraints of the dual are equivalent to $\sum_{p,l}H^{(ij)}_{ps}\mathbf{Y}_{ps}=0$.

\section{Limitations of Moitra's Algorithm} \label{section:Moitra}

Moitra's algorithm can only find NMFs of a particular form, referred to as \textit{stable}. Any NMF can be transformed to a stable NMF without changing the inner dimension, as shown in Lemma~2.6 of \cite{moitra}, but there's no guarantee that the ranks of the factors won't change. The following terminology is needed to define stable NMFs.

\begin{definition}[Admissible Set]
An \textbf{admissible set} of a vector $\bm{v}$ is a set of vectors whose conic hull contains $\bm{v}$.
\end{definition}

\begin{definition}[Lexicographic Ordering]
    Given two sets $\mathbb{S} = \{s_1, s_2, \dots, s_k\}$ and $\mathbb{T} = \{t_1, t_2, \dots, t_k\}$ of the same size such that 
    $s_1 \leq s_2 \leq \dots \leq s_k$ and $t_1 \leq t_2 \leq \dots \leq t_k$,
    we say that $\mathbb{S}$ is lexicographically before $\mathbb{T}$ if there is an $i \in [k]$ such that $s_j \leq t_j$ for all $j=1,2,\dots,i-1$ and $s_i < t_i$.
    Moreover, if $\mathbb{S}$ and $\mathbb{T}$ have different sizes, and $|\mathbb{S}| < |\mathbb{T}|$, then we define $\mathbb{S}$ to be lexicographically before $\mathbb{T}$.
\end{definition}
The important property of lexicographic ordering that we will use is that a smaller set of always lexicographically before a larger set. Now, stability can be defined.
\begin{definition}
A nonnegative matrix factorization \( \mathbf{C} = \mathbf{RE} \) is called \textbf{stable} if:
\begin{enumerate}
    \item For each column \( \mathbf{C}_{:i} \), the support of \( \mathbf{E}_{:i} \) is the lexicographically first admissible subset (of columns of \( \mathbf{R} \)) for \( \mathbf{C}_{:i} \).
    \item For each row \( \mathbf{C}_{j,:} \), the support of \( \mathbf{R}_{j,:} \) is the lexicographically first admissible subset (of rows of \( \mathbf{E} \)) for \( \mathbf{C}_{j,:} \).
\end{enumerate}
\end{definition}
Now, we will prove that restricting search of factorizations to stable NMFs prevents the algorithm from finding the ENNR in general.
\begin{lemma}
    Not every nonnegative matrix that admits an ENMF admits a stable ENMF.
\end{lemma}
\begin{proof}
Consider the COPE $\mathbf{C}^1$ as defined in Eq.~\eqref{C1}, where the factors $A_{\mathcal{U}}$ and $B_{\mathcal{I}}$ correspond to 2-dimensional polytopes $\mathcal{A_{\mathcal{U}}}$ and $\mathcal{B_{\mathcal{I}}}$ respectively, as defined in Section~\ref{Ontic gap}. Construct a new COPE $\mathbf{C}^2$ by keeping the inner polytope $\mathcal{B}_{\mathcal{I}}$ the same and constructing a new outer polytope $\mathcal{A}'_{\mathcal{U}}$. $\mathcal{A}'_{\mathcal{U}}$ is given by adding an extra vertex to $\mathcal{A}_{\mathcal{U}}$ a distance $\epsilon$ from one of the existing vertices, such that this new vertex is convexly independent of the others.
Then, $\mathbf{C}^2$ is the slack of the vertices of the inner polytope $\mathcal{B}_{\mathcal{I}}$ with respect to the facets of the outer polyope $\mathcal{A}'_{\mathcal{U}}$.
We can see that if $\epsilon$ is sufficiently small, an ENMF of inner dimension 4 is still not possible. This is because the amount of extra space that a nested polytope $\mathcal{G}$ of the form in Lemma~\ref{lemma:lemma4} can occupy would be arbitrarily small if $\epsilon$ is arbitrarily small. Therefore, for any such $\mathcal{G}$, the objective value of the linear program used in Lemma~\ref{lemma:algorithm} would change by an arbitrarily small amount. Now, let $\epsilon$ be sufficiently small to keep the ENNR at 5.

Using this construction, $\mathbf{C}^2$ has six rows and $\mathcal{A_I}'$ has six facets. Then, any vertex of a nested polytope $\mathcal{G}$ of the form in Lemma~\ref{lemma:lemma4} is a nonzero distance from at least 4 of the facets of $\mathcal{A_I}'$.
Now, rescale $\mathbf{C}^2$ to obtain a row-stochastic nonnegative matrix $\mathbf{C'}^2$. For an ENMF $\mathbf{C}'^2=\mathbf{R'E'}$ of inner dimension $k\geq 5$, where the factors are, without loss of generality, row-stochastic.
%, the factor $\mathbf{R'}$ has at least four positive entries per column.
In total, $\mathbf{R'}$ has $6k$ entries and at least $4k$ positive entries. Assuming each row of $\mathbf{R}$ has at most three positive entries, $\mathbf{R}$ would have 18 positive entries in total, but this is not possible for $k\geq 5$. Therefore, at least two columns of $\mathbf{E}$ have four supports. Since the rows of $\mathbf{C}'^2$ can be written as convex combinations of the rows of $\mathbf{E'}$, the rows of $\mathbf{E'}$ form a 2-dimensional polytope that contains the columns of $\mathbf{C}'^2$. Therefore, any row of $\mathbf{C}'^2$ can be written as the convex combination of only three rows of $\mathbf{E'}$. This means that if the ENMF is stable, each row of $\mathbf{R'}$ should have only three supports, which is not possible.
\end{proof}

\section{Size Comparison of Inner and Outer Pentagon}\label{size comparison}
The polytope $\mathcal{I}$ has an affine dimension of 2 and a its centre is the point $\mathbf{p}$.
We will show that when the simplex $\mathcal{U}$ is projected onto the affine span of $\mathcal{I}$, it forms a similar pentagon with the same centre and orientation, and only differs in size.

Let $M^1$ be a matrix whose columns are the displacements of the vertices $\{y^i\}$ of $\mathcal{I}$ from the centre $\mathbf{p}$ in a cyclic order of the vertices.
Now, the affine span of the columns of $M^1$ equals their span.
Using the same order of vertices, let $M^2$ be the matrix whose $i$th column is the difference of the two adjacent vertices of $y^i$, e.g., $M^2_{:2}=M^1_{:3}-M^1_{:1}$. Also, the matrix $\frac{\sqrt{5}-1}{2}\mathds{1}_{5\times5}$ has columns equal to the vertices of the simplex $\mathcal{U}$. Define the matrix $M^3$ as $M^3_{;,i}=(\mathds{1}_{5\times5})_{:i}-\mathbf{p}$. The first columns of these matrices are shown below using a particular choice of the order of vertices:
\begin{equation}
\begin{split}
M^1_{:1} &= \frac{\sqrt{5}-1}{10} \begin{pmatrix}
\sin\left(\frac{4\pi}{5}\right) \tan\left(\frac{\pi}{5}\right) - \cos\left(\frac{4\pi}{5}\right) \\
\sin\left(\frac{6\pi}{5}\right) \tan\left(\frac{\pi}{5}\right) - \cos\left(\frac{6\pi}{5}\right) \\
\sin\left(\frac{8\pi}{5}\right) \tan\left(\frac{\pi}{5}\right) - \cos\left(\frac{8\pi}{5}\right) \\
-1 \\
\sin\left(\frac{2\pi}{5}\right) \tan\left(\frac{\pi}{5}\right) - \cos\left(\frac{2\pi}{5}\right)
\end{pmatrix} \\
M^2_{:1} &= \frac{\sqrt{5}-1}{10} \begin{pmatrix}
(\sin\left(\frac{6\pi}{5}\right)-\sin\left(\frac{2\pi}{5}\right)) \tan\left(\frac{\pi}{5}\right) - \cos\left(\frac{6\pi}{5}\right)+\cos\left(\frac{2\pi}{5}\right) \\
(\sin\left(\frac{8\pi}{5}\right)-\sin\left(\frac{4\pi}{5}\right)) \tan\left(\frac{\pi}{5}\right) - \cos\left(\frac{8\pi}{5}\right)+\cos\left(\frac{4\pi}{5}\right) \\
-1-\sin\left(\frac{6\pi}{5}\right) \tan\left(\frac{\pi}{5}\right) + \cos\left(\frac{6\pi}{5}\right) \\
(\sin\left(\frac{2\pi}{5}\right)-\sin\left(\frac{8\pi}{5}\right)) \tan\left(\frac{\pi}{5}\right) - \cos\left(\frac{2\pi}{5}\right)+\cos\left(\frac{8\pi}{5}\right) \\
\sin\left(\frac{4\pi}{5}\right) \tan\left(\frac{\pi}{5}\right) - \cos\left(\frac{4\pi}{5}\right)+1
\end{pmatrix} \\
M^3_{:1} &= \frac{\sqrt{5}-1}{10} \begin{pmatrix}
4 \\ -1 \\ -1 \\ -1 \\ -1 \\
\end{pmatrix}.
\end{split}
\end{equation}
Due to rotational symmetry, each column of $M^1$ is a cyclic permutation of the first column, specifically, each entry is shifted up by one position to get the next column.
Due to the definition of $M^2$, it has the same permutation symmetry and it is clear that $M^3$ also does.

Given that $(M^2_{:1})^\top M^3_{:1}=0$, then by the permutation symmetry, $(M^2_{:i})^\top M^3_{:i}=0$ for all $i$. Given that the columns of $M^1$ are the vertices of a regular pentagon, this implies that when the columns of $M^3$ are projected onto the span of the columns of $M^1$, $M^3_{:i}$ has no component orthogonal to $M^1_{:i}$. Also, we find that $(M^1_{:i})^\top M^3_{:i}$ is positive, so $M^1_{:i}$ is the closest column of $M^1$ to $M^3_{:i}$. Therefore, we can conclude that $\mathcal{B}_\mathcal{I}$ and $\mathcal{P}_{\mathcal{U}}$ are similar pentagons with the same orientation.
The projection of $M^3_{:1}$ onto the span of the columns of $M^1$ has no component in the direction orthogonal to $M^1_{:1}$, so the scale factor can be calculated as $\frac{(M^1_{:1})^\top M^1_{:1}}{(M^1_{:1})^\top M^3_{:1}}$, which simplifies to $\frac{\sqrt{5}-1}{2}\approx0.618$.
\section{Explicit Calculation}
The matrix $G^{2}$ is given by 
\begin{equation}
    \left(
\begin{array}{cccccc}
 \sqrt{5-2 \sqrt{5}} & \frac{1}{2} \sqrt{\frac{1}{2} \left(\sqrt{5}+5\right)} & \frac{1}{4} \sqrt{10-2 \sqrt{5}} & 0 & -\sqrt{\frac{1}{2} \left(5-\sqrt{5}\right)} & -\sqrt{5-2 \sqrt{5}} \\
 -1 & \frac{1}{4} \left(1-\sqrt{5}\right) & \frac{1}{4} \left(\sqrt{5}+1\right) & \sqrt{5}-1 & \frac{1}{2} \left(3-\sqrt{5}\right) & -1 \\
 \frac{5-\sqrt{5}}{10} & \frac{5-\sqrt{5}}{10} & \frac{5-\sqrt{5}}{10} & \frac{5-\sqrt{5}}{10} & \frac{5-\sqrt{5}}{10} & \frac{5-\sqrt{5}}{10} \\
\end{array}
\right),
\end{equation}
and the matrix $B_{\mathcal
{I}}$ is given by 
\begin{equation}
\left(
\begin{array}{ccccc}
 \sqrt{\frac{1}{2} \left(25-11 \sqrt{5}\right)} & \sqrt{5-2 \sqrt{5}} & 0 & -\sqrt{5-2 \sqrt{5}} & -\sqrt{\frac{1}{2} \left(25-11 \sqrt{5}\right)} \\
 \frac{1}{2} \left(1-\sqrt{5}\right) & \sqrt{5}-2 & 3-\sqrt{5} & \sqrt{5}-2 & \frac{1}{2} \left(1-\sqrt{5}\right) \\
 \frac{5-\sqrt{5}}{10} & \frac{5-\sqrt{5}}{10} & \frac{5-\sqrt{5}}{10} & \frac{5-\sqrt{5}}{10} & \frac{5-\sqrt{5}}{10}
\end{array}
\right).   
\end{equation} The matrix $\bar{E}^2$ is given by
\begin{equation}
    \left(
\begin{array}{cccccc}
 \frac{3 \left(632 \sqrt{5}-463\right)}{7676} & \frac{11425-7 \sqrt{5}}{38380} & \frac{11805-3427 \sqrt{5}}{38380} & \frac{528 \sqrt{5}-1237}{7676} & \frac{5400-2443 \sqrt{5}}{38380} & \frac{4576-\frac{6243}{\sqrt{5}}}{7676} \\
 \frac{1}{76} \left(42-11 \sqrt{5}\right) & \frac{3}{380} \left(11 \sqrt{5}+15\right) & \frac{3}{380} \left(11 \sqrt{5}+15\right) & \frac{1}{76} \left(42-11 \sqrt{5}\right) & \frac{1}{380} \left(22 \sqrt{5}-65\right) & \frac{1}{380} \left(22 \sqrt{5}-65\right) \\
 \frac{528 \sqrt{5}-1237}{7676} & \frac{11805-3427 \sqrt{5}}{38380} & \frac{11425-7 \sqrt{5}}{38380} & \frac{3 \left(632 \sqrt{5}-463\right)}{7676} & \frac{4576-\frac{6243}{\sqrt{5}}}{7676} & \frac{5400-2443 \sqrt{5}}{38380} \\
 \frac{-24 \sqrt{5}-31}{7676} & \frac{9 \left(27 \sqrt{5}-205\right)}{38380} & \frac{6895-1657 \sqrt{5}}{38380} & \frac{7 \left(495-112 \sqrt{5}\right)}{7676} & \frac{13 \left(619 \sqrt{5}-80\right)}{38380} & \frac{3440-\frac{2593}{\sqrt{5}}}{7676} \\
 \frac{7 \left(495-112 \sqrt{5}\right)}{7676} & \frac{6895-1657 \sqrt{5}}{38380} & \frac{9 \left(27 \sqrt{5}-205\right)}{38380} & \frac{-24 \sqrt{5}-31}{7676} & \frac{3440-\frac{2593}{\sqrt{5}}}{7676} & \frac{13 \left(619 \sqrt{5}-80\right)}{38380} \\
\end{array}
\right)^\top.
\end{equation}
An orthonormal basis $\{a^{1},a^{2},a^{3}\}$ for the rows of $G^{2}$ is given by the columns of the following matrix:
\begin{equation}
    \left(
\begin{array}{ccc}
 2 \sqrt{\frac{1}{281} \left(19-4 \sqrt{5}\right)} & \frac{4 \left(2 \sqrt{5}-7\right)}{\sqrt{6727-2784 \sqrt{5}}} & \frac{5}{2} \sqrt{\frac{15 \sqrt{5}+74}{4351}} \\
 \sqrt{\frac{1}{562} \left(29 \sqrt{5}+73\right)} & -3 \sqrt{\frac{5 \left(5329-693 \sqrt{5}\right)}{12998498}} & \sqrt{\frac{1473-113 \sqrt{5}}{8702}} \\
 \sqrt{\frac{2}{37-7 \sqrt{5}}} & \sqrt{\frac{477653 \sqrt{5}}{12998498}+\frac{1072449}{12998498}} & \sqrt{\frac{1473-113 \sqrt{5}}{8702}} \\
 0 & 2 \sqrt{\frac{2623-276 \sqrt{5}}{23129}} & \frac{5}{2} \sqrt{\frac{15 \sqrt{5}+74}{4351}} \\
 -\sqrt{\frac{2}{281} \left(7 \sqrt{5}+37\right)} & 6 \sqrt{\frac{5 \left(4863-1852 \sqrt{5}\right)}{6499249}} & \frac{1}{2} \sqrt{\frac{3906-149 \sqrt{5}}{4351}} \\
 -2 \sqrt{\frac{1}{281} \left(19-4 \sqrt{5}\right)} & -\sqrt{\frac{360510 \sqrt{5}}{6499249}+\frac{829082}{6499249}} & \frac{1}{2} \sqrt{\frac{3906-149 \sqrt{5}}{4351}} \\
\end{array}
\right),
\end{equation}
and an orthonormal basis $\{b^{1},b^{2},b^{3}\}$ for the left kernel of $G^{2}$ is given by the columns of
\begin{equation}
    \left(
\begin{array}{ccc}
 \frac{1}{2} \sqrt{\frac{3 \left(67 \sqrt{5}+1126\right)}{1919}} & 0 & 0 \\
 -\sqrt{\frac{435559 \sqrt{5}}{7472586}+\frac{1261193}{7472586}} & 2 \sqrt{\frac{250-46 \sqrt{5}}{1947}} & 0 \\
 -\sqrt{\frac{21055 \sqrt{5}}{7472586}+\frac{48385}{7472586}} & -\sqrt{\frac{685}{3894}+\frac{1187}{3894 \sqrt{5}}} & \sqrt{\frac{1}{10} \left(5-\sqrt{5}\right)} \\
 \frac{1}{2} \sqrt{\frac{2723614-1051741 \sqrt{5}}{3736293}} & \frac{1}{2} \sqrt{\frac{305-134 \sqrt{5}}{9735}} & -\frac{1}{2} \sqrt{1+\frac{2}{\sqrt{5}}} \\
 \frac{1}{2} \sqrt{\frac{729766-141769 \sqrt{5}}{3736293}} & \frac{1}{2} \sqrt{\frac{15865-427 \sqrt{5}}{19470}} & \frac{1}{2} \sqrt{\frac{1}{10} \left(\sqrt{5}+5\right)} \\
 -\frac{1}{2} \sqrt{\frac{5 \left(459134-22213 \sqrt{5}\right)}{3736293}} & -\sqrt{\frac{1541}{15576}+\frac{3307}{15576 \sqrt{5}}} & -\frac{1}{2} \sqrt{\frac{1}{10} \left(5-\sqrt{5}\right)} \\
\end{array}
\right).
\end{equation}
We have found an instance of $\mathbf{Y}$ given by,
\begin{equation}
    \mathbf{Y^*}=\left(
\begin{array}{cccccc}
0 & 0 & 0 & 5-2\sqrt{5} & \frac{1}{2}(\sqrt{5}-1) & 0 \\
0 & 0 & 0 & 0 & 1 & 1 \\
5-2\sqrt{5} & 0 & 0 & 0 & 0 & \frac{1}{2}(\sqrt{5}-1) \\
3-\sqrt{5} & 1 & \dfrac{1}{4}(3-\sqrt{5}) & 0 & 0 & 0 \\
0 & \frac{1}{4}(3-\sqrt{5}) & 1 & 3-\sqrt{5} & 0 & 0 \\
\end{array}
\right)^\top.
\end{equation}
By plugging in $\{a^{1},a^{2},a^{3}\}$, $\{b^{1},b^{2},b^{3}\}$, $\bar{E}^2$ and $\mathbf{Y^*}$ into the constrains of the dual program, it can be shown that $\mathbf{Y^*}$ is in the feasible region. The dual objective value of $\mathbf{Y^*}$ can be calculated as $\frac{1}{2} \left(7-3 \sqrt{5}\right)\approx0.145898$.
\section{Explicit Models} \label{explicit models}
A noncontextual model $\mathbf{C}^1=\mathbf{A}^1\mathbf{B}^1$ of ontic dimension 5 is given by, 
\begin{equation}
\begin{aligned}
\mathbf{C}^1&=\frac{1}{10}\left(
\begin{array}{ccccc}
 2 \left(\sqrt{5}-1\right) & 8-2 \sqrt{5} & 2 \left(\sqrt{5}-1\right) & 3-\sqrt{5} & 3-\sqrt{5} \\
 8-2 \sqrt{5} & 2 \left(\sqrt{5}-1\right) & 3-\sqrt{5} & 3-\sqrt{5} & 2 \left(\sqrt{5}-1\right) \\
 2 \left(\sqrt{5}-1\right) & 3-\sqrt{5} & 3-\sqrt{5} & 2 \left(\sqrt{5}-1\right) & 8-2 \sqrt{5} \\
 3-\sqrt{5} & 3-\sqrt{5} & 2 \left(\sqrt{5}-1\right) & 8-2 \sqrt{5} & 2 \left(\sqrt{5}-1\right) \\
 3-\sqrt{5} & 2 \left(\sqrt{5}-1\right) & 8-2 \sqrt{5} & 2 \left(\sqrt{5}-1\right) & 3-\sqrt{5} \\
\end{array}
\right), \\
\mathbf{A}^1& = \frac{1}{10}
\left(
\begin{array}{ccccc}
 5-\sqrt{5} & 2 \sqrt{5} & 5-\sqrt{5} & 0 & 0 \\
 0 & 5-\sqrt{5} & 2 \sqrt{5} & 5-\sqrt{5} & 0 \\
 0 & 0 & 5-\sqrt{5} & 2 \sqrt{5} & 5-\sqrt{5} \\
 5-\sqrt{5} & 0 & 0 & 5-\sqrt{5} & 2 \sqrt{5} \\
  2 \sqrt{5} & 5-\sqrt{5} & 0 & 0 & 5-\sqrt{5} \\
\end{array}
\right),
\\
\mathbf{B}^1&=
\frac{1}{10}\left(
\begin{array}{ccccc}
 0 & 5-\sqrt{5} & 2 \sqrt{5} & 5-\sqrt{5} & 0 \\
 5-\sqrt{5} & 2 \sqrt{5} & 5-\sqrt{5} & 0 & 0 \\
 2 \sqrt{5} & 5-\sqrt{5} & 0 & 0 & 5-\sqrt{5} \\
 5-\sqrt{5} & 0 & 0 & 5-\sqrt{5} & 2 \sqrt{5} \\
 0 & 0 & 5-\sqrt{5} & 2 \sqrt{5} & 5-\sqrt{5} \\
\end{array}
\right).
\end{aligned}
\end{equation}
A quantum model $\mathbf{C}^1=A^1B^1$ of dimension 3 (subset of rebit theory) is given by,
\begin{equation}
    \begin{aligned}
        A^1&=\left(
\begin{array}{ccc}
 \frac{1}{5} \sqrt{5-2 \sqrt{5}} & \frac{1}{\sqrt{5}}-\frac{2}{5} & \frac{1}{5} \\
 0 & \frac{3}{5}-\frac{1}{\sqrt{5}} & \frac{1}{5} \\
 -\frac{1}{5} \sqrt{5-2 \sqrt{5}} & \frac{1}{\sqrt{5}}-\frac{2}{5} & \frac{1}{5} \\
 -\frac{1}{10} \sqrt{50-22 \sqrt{5}} & \frac{1}{10} \left(1-\sqrt{5}\right) & \frac{1}{5} \\
 \frac{1}{10} \sqrt{50-22 \sqrt{5}} & \frac{1}{10} \left(1-\sqrt{5}\right) & \frac{1}{5} \\
\end{array}
\right),\\
B^1&=\left(
\begin{array}{ccccc}
 0 & \sqrt{\frac{\sqrt{5}}{8}+\frac{5}{8}} & \sqrt{\frac{5}{8}-\frac{\sqrt{5}}{8}} & -\sqrt{\frac{5}{8}-\frac{\sqrt{5}}{8}} & -\sqrt{\frac{\sqrt{5}}{8}+\frac{5}{8}} \\
 1 & \frac{1}{4} \left(\sqrt{5}-1\right) & \frac{1}{4} \left(-\sqrt{5}-1\right) & \frac{1}{4} \left(-\sqrt{5}-1\right) & \frac{1}{4} \left(\sqrt{5}-1\right) \\
 1 & 1 & 1 & 1 & 1 \\
\end{array}
\right).
    \end{aligned}
\end{equation}
This represents using a single measurement given by the POVM whose operators are of the form $A^1_{i1}\mu_x+A^1_{i2}\mu_z+A^1_{i3}\mathds{1}_{2 \times 2}$ for $i=1,...,5$,
and states given by $\frac{1}{2}\left(B^1_{1j}\mu_x+B^1_{2j}\mu_z+B^1_{3j}\mathds{1}_{2 \times 2}\right)$ for $j=1,...,5$.

An ontological model $\mathbf{C}^1=\mathbf{A}^2\mathbf{B}^2$ of ontic dimension 4 is given by,
\begin{equation}
\begin{aligned}
\mathbf{A}^2&=\left(
\begin{array}{cccc}
 0 & 1-\frac{1}{\sqrt{5}} & 0 & \frac{1}{10} \left(5-\sqrt{5}\right) \\
 \frac{1}{\sqrt{5}} & \frac{1}{\sqrt{5}} & 0 & 0 \\
 1-\frac{1}{\sqrt{5}} & 0 & \frac{1}{10} \left(5-\sqrt{5}\right) & 0 \\
 0 & 0 & 1-\frac{1}{\sqrt{5}} & \frac{1}{10} \left(3 \sqrt{5}-5\right) \\
 0 & 0 & \frac{1}{10} \left(3 \sqrt{5}-5\right) & 1-\frac{1}{\sqrt{5}} \\
\end{array}
\right),\\
\mathbf{B}^2&=\frac{1}{20}\left(
\begin{array}{ccccc}
 20 \left(\frac{2}{\sqrt{5}}-\frac{1}{2}\right) & 5-\sqrt{5} & 0 & 2 \left(3 \sqrt{5}-5\right) & 3 \left(5-\sqrt{5}\right) \\
 20 \left(\frac{2}{\sqrt{5}}-\frac{1}{2}\right) & 3 \left(5-\sqrt{5}\right) & 2 \left(3 \sqrt{5}-5\right) & 0 & 5-\sqrt{5} \\
 20 \left(1-\frac{2}{\sqrt{5}}\right) & 0 & 2 \left(5-\sqrt{5}\right) & 20 \left(1-\frac{1}{\sqrt{5}}\right) & 4 \sqrt{5} \\
 20 \left(1-\frac{2}{\sqrt{5}}\right) & 4 \sqrt{5} & 20 \left(1-\frac{1}{\sqrt{5}}\right) & 2 \left(5-\sqrt{5}\right) & 0 \\
\end{array}
\right).
\end{aligned}
\end{equation}

\end{widetext}
\end{document}